\documentclass[hidelinks, 10pt,conference]{IEEEtran}
\IEEEoverridecommandlockouts

\usepackage{cite}
\usepackage{amsmath,amssymb,amsfonts}
\usepackage{algorithmic}
\usepackage{enumitem}

\usepackage{graphicx}
\usepackage{textcomp}
\usepackage{xcolor}
\def\BibTeX{{\rm B\kern-.05em{\sc i\kern-.025em b}\kern-.08em
    T\kern-.1667em\lower.7ex\hbox{E}\kern-.125em}}

\usepackage{blindtext}
\usepackage{float}

\usepackage{amsthm}

\usepackage{color,soul}
\usepackage{esvect}
\usepackage[linesnumbered,ruled]{algorithm2e}
\usepackage{url}
\usepackage[hidelinks]{hyperref}

\SetArgSty{textup}
\usepackage{bm,xparse}

\SetCommentSty{mycommfont}

\usepackage{booktabs}

\usepackage{subfig}
\usepackage{caption}
\usepackage{mathtools}
\DeclarePairedDelimiter\ceil{\lceil}{\rceil}

\newcommand{\eat}[1]{}
\newenvironment{proof-sketch}{\noindent{\bf Sketch of Proof}\hspace*{1em}}{\qed\bigskip}

\newcommand{\thickbar}[1]{\mathbf{\bar{\text{$#1$}}}}

\newtheorem{definition}{Definition}
\newtheorem{theorem}{Theorem}

\newtheorem{lemma}{Lemma}

\newtheorem{example}{Example}

\usepackage{pict2e}

\usepackage{pgfplots}
\usepackage{pgfplotstable}

\definecolor{bblue}{HTML}{4F81BD}
\definecolor{rred}{HTML}{C0504D}
\definecolor{ggreen}{HTML}{9BBB59}
\definecolor{ppurple}{HTML}{9F4C7C}

\begin{document}

\title{\LARGE Composite Resource Scheduling for Networked Control Systems}

\author{\IEEEauthorblockN{Peng Wu\IEEEauthorrefmark{1}, Chenchen Fu\IEEEauthorrefmark{2}, Tianyu Wang\IEEEauthorrefmark{1}, Minming Li\IEEEauthorrefmark{3}, Yingchao Zhao\IEEEauthorrefmark{4}, Chun Jason Xue\IEEEauthorrefmark{3}, Song Han\IEEEauthorrefmark{1}}
\IEEEauthorblockA{\IEEEauthorrefmark{1}University of Connecticut \; \; \{peng.wu, tianyu.wang, song.han\}@uconn.edu}  
\IEEEauthorblockA{\IEEEauthorrefmark{2}Southeast University \; \;
chenchen\_fu@seu.edu.cn}
\IEEEauthorblockA{\IEEEauthorrefmark{3}City University of Hong Kong \; \;
\{minming.li, jasonxue\}@cityu.edu.hk}
\IEEEauthorblockA{\IEEEauthorrefmark{4}Hong Kong Caritas Institute of Higher Education \; \; zhaoyingchao@gmail.com}
}

\maketitle
\thispagestyle{plain}
\pagestyle{plain}

\begin{abstract}
Real-time end-to-end task scheduling in networked control systems (NCSs) requires the joint consideration of both network and computing resources to guarantee the desired quality of service (QoS). This paper introduces a new model for composite resource scheduling (CRS) in real-time networked control systems, which considers a strict execution order of sensing, computing, and actuating segments based on the control loop of the target NCS. We prove that the general CRS problem is NP-hard and study two special cases of the CRS problem.
The first case restricts the computing and actuating segments to have unit-size execution time while the second case assumes that both sensing and actuating segments have unit-size execution time. We propose an optimal algorithm to solve the first case by checking the intervals with 100\% network resource utilization and modify the deadlines of the tasks within those intervals to prune the search. For the second case, we propose another optimal algorithm based on a novel backtracking strategy to check the time intervals with the network resource utilization larger than 100\% and modify the timing parameters of tasks based on these intervals. For the general case, we design a greedy strategy to modify the timing parameters of both network segments and computing segments within the time intervals that have network and computing resource utilization larger than 100\%, respectively. The correctness and effectiveness of the proposed algorithms are verified through extensive experiments.
\end{abstract}

\section{Introduction}\label{sec:introduction}

Networked control systems (NCSs) are fundamental to many mission- and safety-critical applications that must work under real-time constraints to ensure timely collection of sensor data and on-time delivery of control decisions. The Quality of Service (QoS) offered by a NCS is thus often measured by how well it satisfies the end-to-end deadlines of the real-time tasks executed in NCSs~\cite{hei_pipac:_2013, gatsis_optimal_2014,tian_deadline-constrained_2016,borgers_tradeoffs_2017,zhan_optimal_2016}. 

A typical real-time task in NCSs involves both computing component(s) for executing the control algorithms on the controller and communication component(s) for exchanging sensor data and control signals between sensors/actuators and the controller. Traditional approaches to scheduling those tasks consider CPU and network resource scheduling separately. They either make an oversimplified assumption that the execution time of the control algorithm is negligible or consider it as constant. For example, extensive work has been reported in recent years on sensing and control task scheduling in real-time industrial networks~\cite{han_reliable_2011,leng_improving_2014,saifullah_real-time_2010,lu_real-time_2016,yang_assignment_2015,dezfouli_mobility-aware_2016,craciunas_scheduling_2016,nayak_time-sensitive_2016,lin_tsca_2018}. Most of those work, however only focused on modeling the constraints of transmission conflicts but paid less attention to the computation time of the control algorithms when enforcing the end-to-end deadlines. Those approaches thus either cannot handle complex NCS applications which require the implementation of time-consuming control algorithms such as model predictive control using online optimization, data-driven system identification, and online learning control~\cite{jo_development_2014,mohammadi_optimizing_2018,jo_development_2015,klus_data-driven_2020,lee_nonparametric_2017}, or will lead to unnecessarily low utilization of the computing and network resources due to the lack of efficient methods to schedule those resources in a joint fashion.  Take the F1/10 autonomous car system as an example that aims to race in a rectangular track while avoiding any obstacles~\cite{chen_online_2019}. The communication component on the system for exchanging sensing/actuating information has a worst-case execution time of $0.3$ milliseconds; while the PID controller for steering control and the vision controller for identifying corners have worst-case execution times of $0.4$ and $50$ milliseconds, respectively. These comparable execution times of both networking and computing tasks motivate us to introduce a new task model, namely the Composite Resource Scheduling (CRS) model, for jointly scheduling network and computing resources in NCSs. 
In CRS, each real-time composite task consists of three consecutive and dependent segments: a sensing segment to transmit sensor data to the controller, a computing segment to compute the control decisions, and an actuating segment to transmit the control signals to the actuators. The sensing/actuating segments together are called network segments. They share the same network resource while the computing segments compete for the computing resource on the controller. 

Similar models in the literature include 
PRedictable Execution Model (PREM)~\cite{pellizzoni_predictable_2011,alhammad_time-predictable_2014,yao_global_2016} and Acquisition Execution Restitution (AER) model~\cite{maia_closer_2016}. 
These two models, however are specifically designed to schedule resources in multi-core systems with shared memory where the communication phases including memory read and write are considered as non-preemptive, which is reasonable for the design of multi-core systems. By contrast, the network segments of the CRS model for NCSs are designed to be preemptive to meet the flexibility of the network scheduling. This provides better parallelism when scheduling sensing/actuating segments and computing segments from different tasks. \eat{For a detailed discussion of the related work, please refer to our technical report~\cite{technical_report_2021}.}


Our work attempts to tackle the composite resource scheduling problem in NCSs, which aims to optimize the usage of network and computing resources in NCSs under the end-to-end deadline constraints. Based on the proposed CRS model, we provide a comprehensive analysis on the complexity of the problem, and prove that the general CRS problem is NP-hard in the strong sense. We thus in this paper study two special cases of the CRS problem and present a greedy heuristic for the general cases as well, focusing on the design of the corresponding algorithms to construct feasible composite schedules. We start with the first case where both the computing segment and actuating segment have unit-size execution time. In this case, an optimal scheduling algorithm is proposed. The algorithm first modifies the deadlines of the CRS task set based on the intervals with 100\% network resource utilization and then applies the Earliest Deadline First (EDF) scheduling algorithm to schedule the sensing and actuating segments together according to the modified deadlines. For the second case where the execution time of both sensing and actuating segments are unit-size, and the execution time of the computing segment can be an arbitrary integer larger than one, we design another optimal scheduling algorithm which schedules the computing segments in the first stage and the sensing and actuating segments in the second stage. If the scheduling fails in the second stage, the algorithm rolls back to the initial stage based on a novel backtracking search strategy by adding new constraints to the original problem recursively, and eventually lead to a feasible composite schedule if it exists. For the general case, we propose a heuristic solution with a roll-back mechanism. EDF is first employed to schedule the segments. If EDF fails to find a feasible schedule, we locate the intervals with either network resource or computing resource utilization larger than 100\% and modify the deadline of the segment included in the located interval with the earliest release time. We iteratively run EDF and modify the timing parameters until we find a feasible schedule. The effectiveness of the proposed algorithms has been validated through extensive experiments. Our results show that the proposed scheduling algorithms outperform the baseline algorithms in terms of schedulability.

The remainder of this paper is organized as follows: Section~\ref{sec:model} presents the task model and the formulation of the general composite resource scheduling problem and its two special cases. Section~\ref{sec:CRS-111} and Section~\ref{sec:CRS-1m1} develop the optimal scheduling algorithms for the first and second cases, respectively. Section~\ref{sec:CRS-1m1} introduces a heuristic scheduling algorithm for the general case. Section~\ref{sec:evaluation} presents the experimental results. Section~\ref{sec:relatedWorks} gives a summary of the related works. Section~\ref{sec:conclusion} concludes the paper and discusses the future work.

\section{Task Model and Problem Formulation}\label{sec:model}

In this section, we introduce the CRS model and present the constraint programming formulation of the composite resource scheduling problem. Based on different settings on the execution time of the network/computing segments, we introduce three CRS models.
\eat{
\begin{figure}[!ht]
    \centering
    \includegraphics[width=3.0in]{fig_robot.png}
    \label{fig:loop}
\end{figure}
\begin{figure}[!ht]
    \centering
    \includegraphics[width=3.0in]{fig_smart.png}
    \label{fig:loop}
\end{figure}

}

\subsection{Task Model}

Consider a NCS configured in a star topology. The sensors and actuators form the leaf nodes while the controller runs on the root (Gateway). To ensure efficient and safe operation of the control system, the transmissions of sensor data/control signals and the execution of the control algorithm should follow a strict order and meet the designated end-to-end timing requirement. In a typical real-time composite task in NCSs, one or multiple sensors send their measurements to the controller first. The controller then computes the control signals and delivers them to one or multiple actuators. In practice, while the CPU is waiting for the sensing data, it can be set to idle. Once the transmission of sensing data is finished, an interrupt notifies the CPU to obtain the data stored in a specified shared memory. This decouples the computation from sensing and actuating so that the transmissions of sensing data and control signals share the network resource while the executions of the control algorithms only compete for the CPU resource. In our CRS model, we consider multiple-input multiple-output (MIMO) control systems, which have multiple sensors as the input and multiple actuators as the output. 
As shown in Fig.~\ref{fig:composite}, each real-time composite task $\tau_i = (T_i, D_i, C_{i,s}, C_{i,c}, C_{i,a})$ is associated with a period $T_i$ and a relative deadline $D_i$, where it holds that $D_i \leq T_i$. It consists of three consecutive and dependent segments:
\begin{itemize}
\item Sensing segment: it utilizes the network resource and has the transmission time of $C_{i,s}$;  
\item Computing segment: it utilizes the CPU resource and has the execution time of $C_{i,c}$;
\item Actuating segment: it utilizes the network resource and has the transmission time of $C_{i,a}$. 
\end{itemize}

\begin{figure}
  \centering
  \includegraphics[width=\columnwidth]{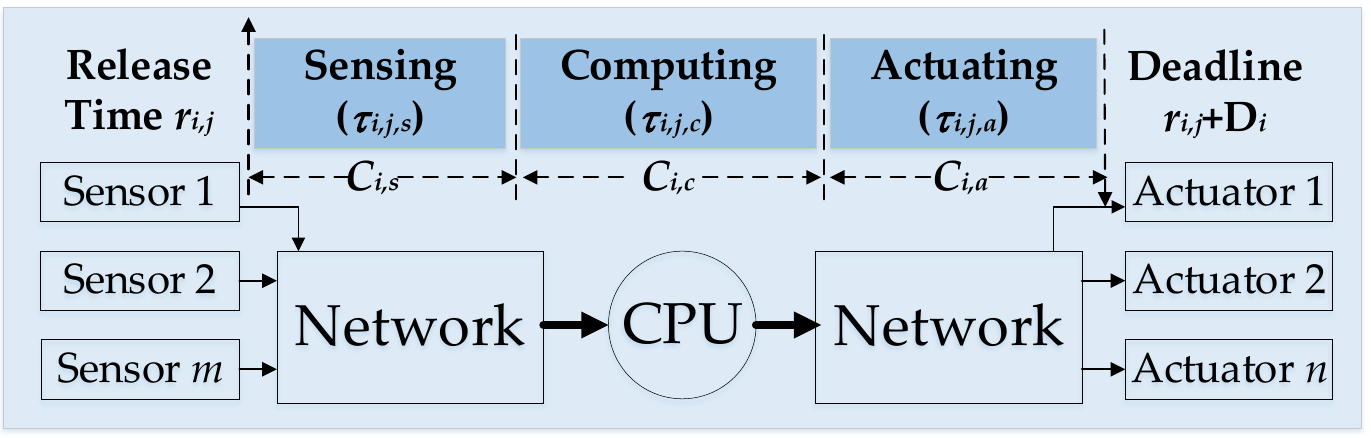}
  \caption{An overview of the real-time composite task model}
  \vspace{0.1in}
  \label{fig:composite}
\end{figure}

In this paper, we consider a time-slotted system for network. A network segment consists of multiple fragments where each fragment takes exactly one time slot in the super-frame to transmit. We define the execution times $C_{i,s}$ and $C_{i,a}$ of the sensing and actuating segments to be the sum of the network transmission times of all sensors and actuators connected to a control task, and that their activation frequency is equal to the one of the task. We assume that (1) the unit sizes of the network and computing resources are the same; (2) the network is single-channel and the CPU is a preemptive uniprocessor; (3) the task system is a synchronous system; and (4) the release time, deadline and execution time are integers.

\subsection{Problem Formulation}

Based on the above task model, we now present the formal definition of the composite resource scheduling problem and give its constraint programming formulation.

\vspace{0.05in}
\noindent \textbf{Composite Resource Scheduling (CRS) Problem:} Consider a set of real-time composite tasks \(\{\tau_1, \tau_2,...,\tau_n\}\) with $\tau_i = (T_i, D_i, C_{i,s}, C_{i,c}, C_{i,a})$ and the hyper-period $\mathcal{H}$ obtained by computing the least common multiple of the periods. The objective of the CRS problem is to find a feasible composite schedule with a length of  $\mathcal{H}$ if it exists so that the deadlines of all the real-time composite tasks are met.

We first formulate the CRS problem as a constraint programming problem. It aims to construct the feasible schedules $\mathcal{S}_{{\rm net}}  =\{f_{i,j,s,k}, f_{i,j,a,k}\}$ for the network segments and $\mathcal{S}_{{\rm com}}=\{f_{i,j,c,k}\}$ for the computing segments, where $f_{i,j,s,k}$, $f_{i,j,c,k}$, and $f_{i,j,a,k}$ are the finish times of the $k^{\text{th}}$ unit of sensing/computing/actuating segments of the $j^{\text{th}}$ instance of the task $\tau_{i}$, respectively. The $j^{\text{th}}$ instance of task $\tau_{i}$, denoted as $\tau_{i,j}$, has a release time $r_{i,j}$ and a deadline $r_{i,j} + D_{i}$. The construction of the feasible schedules is subject to the following constraints: 

\vspace{6pt}
\noindent{\bf Release time and deadline constraints:}
\vspace{6pt}
\begin{equation} \label{eq:1}
\begin{split}
&f_{i,j,s,1}\geq r_{i,j} + 1   \\ 
&f_{i,j,a,C_{i,a}}\leq r_{i,j} + D_{i}
\end{split}
\end{equation}

\vspace{6pt}
\noindent{\bf Segment order constraints:}
\begin{equation} \label{eq:2}
\begin{split}
&f_{i,j,c,1}\geq f_{i,j,s,C_{i,s}}+1   \\ 
&f_{i,j,a,1}\geq f_{i,j,c,C_{i,c}}+1 \\
f_{i,j,s,k+1}\geq f_{i,j,s,k}&+1,  \quad   \forall k \in \mathbb{Z}\cap [1, C_{i,s}-1]\\
f_{i,j,c,k+1}\geq f_{i,j,c,k}&+1, \quad   \forall k \in \mathbb{Z}\cap [1, C_{i,c}-1]\\
f_{i,j,a,k+1} \geq f_{i,j,a,k}&+1, \quad   \forall k \in \mathbb{Z}\cap [1, C_{i,a}-1]
\end{split}
\end{equation}

\vspace{3pt}
\noindent  {\bf Network resource constraints:}
\vspace{3pt}
\begin{equation} \label{eq:3}
\begin{split}
    &\forall i\neq z, \quad  \forall x,y \in\{a,s\} \\
    &\forall p \in \mathbb{Z}\cap[1,C_{i,x}],  \quad \forall q \in \mathbb{Z}\cap[1,C_{z,y}] \\
    &f_{i,j,x,p}\leq f_{z,j,y,q}-1 \quad  \text{OR} \quad  f_{i,j,x,p}\geq f_{z,j,y,q}+1  
\end{split}
\end{equation}


\vspace{0.05in}
\noindent {\bf Computing resource constraints:}
\vspace{3pt}
\begin{equation} \label{eq:4}
\begin{split}
&\forall i\neq z,   \quad \forall p \in \mathbb{Z}\cap[1,C_{i,c}],  \quad \forall q \in \mathbb{Z}\cap[1,C_{z,c}]\\
& f_{i,j,c,p}\leq f_{z,j,c,q}-1 \quad  \text{OR} \quad  f_{i,j,c,p}\geq f_{z,j,c,q}+1 
\end{split}
\end{equation}
\vspace{2pt}

Constraint \eqref{eq:1} requires that the first unit of the sensing segment and the last unit of the actuating segment of any task $\tau_i$ are constrained by $\tau_i$'s release time and deadline, respectively. Constraint \eqref{eq:2} defines the constraints for the sequential order for each unit of sensing, computing, and actuating segments. One can easily derive the earliest possible release time and the latest possible deadline for each segment based on Constraints \eqref{eq:1} and \eqref{eq:2}. Constraints \eqref{eq:3} and \eqref{eq:4} define the constraints for different composite tasks to compete for the network and computing resources, where the OR operator in the constraints can be modeled by using the big M method~\cite{griva_linear_2009}. Note that the above formulation employing the time-indexed model~\cite{KU_wenyang_mixedInteger_2016} is computationally efficient since the problem studied in the paper is a preemptive flow shop scheduling problem.

\vspace{0.05in}
\begin{table}[h]
    \centering
    \small
    \caption{Problems and solutions under different CRS models.}
    \begin{tabular}{|c|c|}
    \hline 
    Task Model & Complexity or Solution  \\
        \hline  
        $h\text{-}1\text{-}1$ & Exponential-time solvable (Alg.~\ref{alg1})
         \\
         $1\text{-}m\text{-}1$ & Exponential-time solvable (Alg.~\ref{alg2})
         \\
        $h_1\text{-}h_2\text{-}h_3$ & NP-hard and solved by heuristics (Alg.~\ref{alg4}) 
         \\
    \hline
    \end{tabular}
    \label{tab:table1}
\end{table}
\vspace{0.05in}

\eat{
\begin{center}
\vbox{\centering{\footnotesize TABLE I. \quad Problems and solutions under different CRS models} \vskip2mm
\renewcommand{\baselinestretch}{1.25}

{\footnotesize\centerline{\tabcolsep=20pt\begin{tabular}{0.485\textwidth}{cc}
\toprule
 Task Model & Complexity or Solution  \\
\hline  $h\text{-}1\text{-}1$ & Polynomial-time solvable (Alg.~\ref{alg1})
 \\
 $1\text{-}m\text{-}1$ & Polynomial-time solvable (Alg.~\ref{alg2})
 \\
$h_1\text{-}h_2\text{-}h_3$ & NP-hard and solved by heuristics (Alg.~\ref{alg4}) 
 \\

\bottomrule
\end{tabular}}}}
\vspace{0.1in}
\end{center}
}

To study the CRS problem comprehensively, we consider three cases of the CRS model based on the size of the execution time of each segment. These cases are summarized in Table I. The first case is the $h\text{-}1\text{-}1$ model where \(C_{i,c}=C_{i,a}=1\) and $C_{i,s}=h$ ($h>0$) for each segment of a task. The second case uses the $1\text{-}m\text{-}1$ model with \(C_{i,s}=C_{i,a}=1\) and $C_{i,c}=m$ ($m>1$). The third case is the general model $h_1\text{-}h_2\text{-}h_3$ with \(C_{i,s}=h_1\), \(C_{i,c}=h_2\), and $C_{i,a}=h_3$ ($h_1, h_2, h_3>0$). These three cases represent three types of application scenarios in the real life. While $h-1-1$ refers to the task which has a longer sensing segment, $1-m-1$ represents the task which has a longer computing segment, and $h_1\text{-}h_2\text{-}h_3$ represents the general applications. The following theorem shows that the CRS problem under the general model is NP-hard in the strong sense.

\begin{theorem}
The CRS problem under the general model is NP-hard in the strong sense. 
\end{theorem}

\begin{proof}
We reduce 3-Partition to the CRS problem. An instance of 3-Partition consists of a list $A = (x_1, x_2, . . . , x_{3n})$ of positive integers such that $\sum x_i = nB$,  $\frac{B}{4}<x_i<\frac{B}{2}$ 
for each $1\leq i\leq 3n$, there exists a partition of $A$ into $A_1$, $A_2$,...,$A_n$ such that $\sum_{x_i \in A_k}x_i = B$ for each $1 \leq k \leq n$~\cite{garey_computers_1990}. 

Given an instance $A = (x_1, x_2, . . . , x_{3n})$ of 3-Partition, we construct an instance of the CRS problem
as follows. As shown in Fig.~\ref{fig:NP}, there will be $3n + 1$ tasks. The first $3n$ tasks are {\em partition} tasks. For each $1 \leq i \leq 3n$, the task $\tau_{i}$ satisfies that $C_{i,s} = 2\mu x_i$, $C_{i,c} =  2\mu x_i$, $C_{i,a} = 1$, $D_i = (4n+1)\mu B + 3n$, and $T_i = 4\mu B$, where $B$ is the sum of each partition and  $\mu = \ceil{\frac{3n}{B}}$. The last task is {\em divider} task which satisfies that $T_i = D_i = 4\mu B$ , $C_{i,s} = \mu B$, $C_{i,c} = 2\mu B$, and $C_{i,a} = \mu B$. The $n+1$ divider tasks divide the timeline of $[\mu B, (4n+3)\mu B]$ into $n+1$ intervals of full computing resource utilization interleaved with $n$ intervals of full network resource utilization, where the partition tasks are scheduled. The length of each of these intervals is exactly $2\mu B$. The sensing segments of the partition tasks from the same partition must be scheduled in a length $2\mu B$ of the interval $[(2j-1)\mu B, (2j+1)\mu B]$ so that their corresponding computing segments 
can anticipate releasing exactly at or before $(2j+1)\mu B$ and fully utilizing the interval $[(2j+1)\mu B, (2j+3)\mu B]$, where $1 \leq j \leq n$. Thus, it is easy to see that there is a feasible schedule of the CRS problem if and only if there exists a 3-Partition. It is clear that the above reduction is polynomial. We thus have proved that the CRS problem is NP-hard in the strong sense.
\end{proof}

\begin{figure}
  \centering
  \includegraphics[width=\columnwidth]{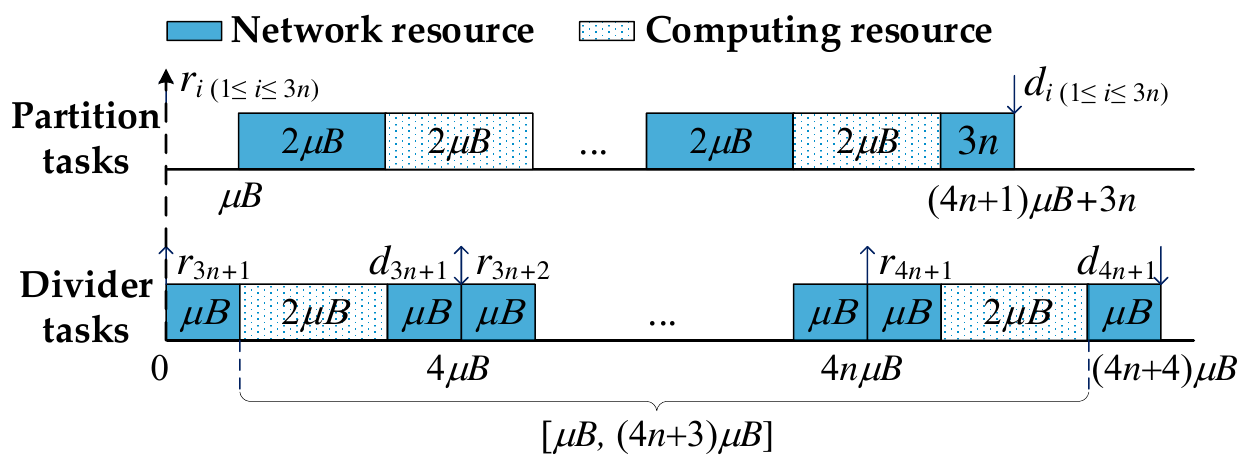}
  \caption{An instance of the composite schedule}
  \vspace{0.1in}
  \label{fig:NP}
\end{figure}

We first exploit the earliest possible release time and the latest possible deadline to define the \textit{effective timing parameters} of the network and computing segments.

\begin{definition}
\label{def:vir}
{\bf Effective Release Time/Deadline:} Given an instance of  real-time composite task $\tau_{i,j}$ with the release time $r_{i,j}$ and deadline $d_{i,j}$, the effective release time/deadline of its sensing segment $\tau_{i,j,s}$ are defined as $\thickbar{r}_{i,j,s} = r_{i,j}$ and $\thickbar{d}_{i,j,s} = d_{i,j}-C_{i,a}-C_{i,c}$. The effective release time/deadline of its computing segment $\tau_{i,j,c}$ are defined as $\thickbar{r}_{i,j,c}=r_{i,j}+C_{i,s}$ and $\thickbar{d}_{i,j,c}=d_{i,j}-C_{i,a}$. The effective release time/deadline of its actuating segment $\tau_{i,a}$ are defined as $\thickbar{r}_{i,j,a} = r_{i,j}+C_{i,s}+C_{i,c}$ and $\thickbar{d}_{i,j,a} = d_{i,j}$.
\end{definition}

We say that a network segment is \textit{effectively included} in a time interval if its effective release time and effective deadline are both within that interval. For instance, a network segment $\tau_{i,j,{\rm net}}$ is effectively included in $[t_0, t_1]$ if $\thickbar{r}_{i,j,{\rm net}}\geq t_0$ and $\thickbar{d}_{i,j,{\rm net}} \leq t_1$. Based on this condition, we define an \textit{effective network demand} over a given interval to be the sum of the execution time of all network segments that are effectively included in that interval. We define the \textit{effective network overload interval} and \textit{effective network tight interval} as follows.

\begin{definition}
\label{def:tig}
{\bf Effective Network Overload/Tight Interval (ENOI/ENTI):} Given a set of real-time composite tasks and a time interval $[t_0, t_1]$, $[t_0, t_1]$ is an effective network overload interval if the effective network demand over $[t_0, t_1]$ is larger than $t_1-t_0$. $[t_0, t_1]$ is an effective network tight interval if the effective network demand over $[t_0, t_1]$ is equal to $t_1-t_0$.
\end{definition}

\begin{figure}[htb]
  \centering
  \includegraphics[width=0.82\columnwidth]{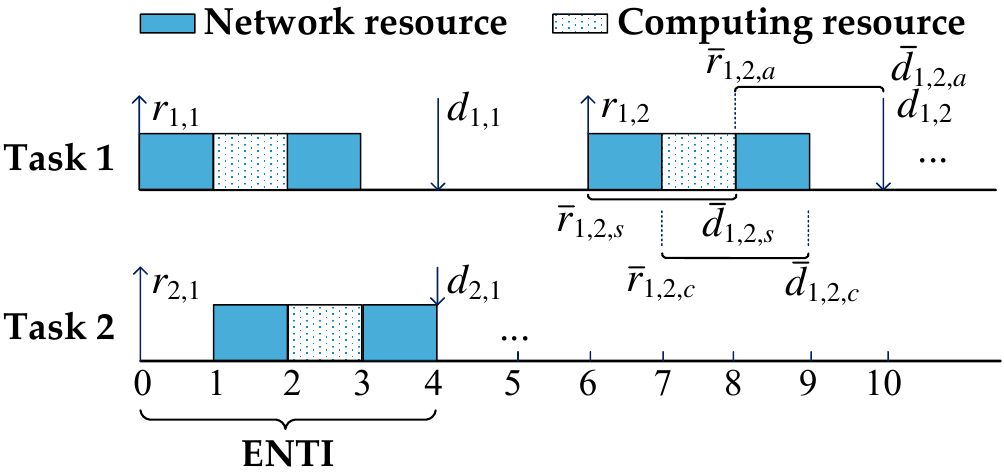}
  \caption{An instance of ENTI}
  \vspace{0.1in}
  \label{fig:eg1}
\end{figure}

\begin{example}
Consider two tasks $\tau_1 = \tau_2 = (6, 4, 1, 1, 1)$. As shown in Fig.~\ref{fig:eg1}, the instance $\tau_{1,2}$ satisfies $r_{1,2} = 6$ and $d_{1,2} = 10$. The effective release times and deadlines of its  three segments are $\thickbar{r}_{1,2,s} = 6$,  $\thickbar{d}_{1,2,s} = 8$, $\thickbar{r}_{1,2,c} = 7$, $\thickbar{d}_{1,2,c} = 9$, $\thickbar{r}_{1,2,a} = 8$, and $\thickbar{d}_{1,2,a} = 10$. The effective network demand over $[0, 4]$ is $4$, so the interval $[0, 4]$ is an ENTI.   
\end{example}

Since the CRS problem under the general model is NP-hard, in the following we first focus on the design of exact algorithms for the first two simple models and then present an effective heuristic solution for the general model.

\eat{\begin{definition} 
\label{def:EDF}
{\bf Earliest Deadline First:} Given a set of real-time composite tasks \(\{\tau_1, \tau_2,..., \tau_n\}\), Earliest Deadline First (EDF) executes the task with the earliest absolute deadline among all the ready tasks, breaking deadline ties according to the shortest remaining processing time.
\end{definition}}

\section{CRS Problem under $H\text{-}1\text{-}1$ Model}\label{sec:CRS-111}

The $h\text{-}1\text{-}1$ model represents a wide range of NCSs which need to transport a large amount of sensing data while the required computing and actuation time are short.  In the following, after giving several important observations, we present an exact optimal algorithm to solve it.




Based on the effective timing parameters, we employ EDF to schedule the real-time composite tasks in the following fashion. We utilize a network ready queue and a computing ready queue for scheduling the network segments and computing segments, respectively. Once a sensing segment $\tau_{i,j,s}$ is finished at $f_{i,j,s}$, its computing segment $\tau_{i,j,c}$ is released at $f_{i,j,s}$. Similarly, if a computing segment $\tau_{i,j,c}$ is finished at $f_{i,j,c}$, its actuating segment $\tau_{i,j,c}$ is released at $f_{i,j,c}$. The actuating segments will be scheduled together with all other network segments (including both sensing and actuating segments from other task instances) using EDF. The effective deadlines of the sensing and actuating segments are used to decide their priorities. Ties are broken by scheduling the sensing segments first. For the same type of segments, ties are broken by scheduling the segment with the least  laxity.

\begin{lemma}
\label{lem:3_stage_m_1_1}
Given a real-time composite task set $\mathcal{T}$ under the $h\text{-}1\text{-}1$ model, if EDF can find a feasible network schedule for the sensing segments based on their effective deadlines, it can also find a feasible computing schedule based on the effective deadlines of their computing segments.     
\end{lemma}

It is straightforward to prove the correctness of Lemma~\ref{lem:3_stage_m_1_1}. As the finish times of the sensing segments are different, the release times of their corresponding computing segments are different as well. Given that each computing segment takes unit-size execution time, we can always schedule the computing segments immediately after the completion of its corresponding sensing segments. Lemma~\ref{lem:3_stage_m_1_1} indicates that scheduling the sensing segments  does not interfere with scheduling the computing segments. However, it is hard to guarantee that EDF can also construct a feasible schedule for the actuating segments as they compete for the network resource with the sensing segments. Therefore, in the following we focus on adapting EDF to schedule the sensing and actuating segments.

The following lemma gives a necessary condition to determine the schedulability of the real-time composite task set under the $h\text{-}1\text{-}1$ model. 

\begin{lemma}
\label{lem:schedulability_m_1_1}
Given a real-time composite task set $\mathcal{T}$ under the $h\text{-}1\text{-}1$ model, if there exists an effective network overload interval (ENOI), the task set is unschedulable.     
\end{lemma}
\begin{proof}
The effective deadline and effective release time of a network segment correspond to that segment's latest possible deadline and earliest possible release time. If there exists an ENOI, it indicates that there is no sufficient network resource to schedule those network segments within that interval. Thus the task set is not schedulable.
\end{proof}

An effective network tight interval (ENTI) indicates that this interval has a 100\% network resource utilization in any feasible schedule. Thus, it is important to utilize the ENTIs to move up the deadline of the tasks whose deadlines are within the interval but their network segments are not effectively included in the interval.

The following definition and lemma show that EDF can find the feasible schedule if there exists one after applying the ENTIs to modify the task deadlines.

\begin{definition}
\label{def:fail}
{\bf Continuous Interval:} Given a real-time composite task set $\mathcal{T}$ under the $h\text{-}1\text{-}1$ model, suppose $\tau_{z,j,{\rm net}}$ is the first network segment scheduled by EDF to miss its deadline $\thickbar{d}_{z,j,{\rm net}}$ and $b_0$ is the start time of a network segment in the schedule, we define $[b_0, \thickbar{d}_{z,j,{\rm net}}]$ to be a continuous interval if it is fully utilized by the network segments and there exists no network segment scheduled in the interval $[b_0-1, b_0]$ or the network segment scheduled in $[b_0-1, b_0]$ has a deadline later than $\thickbar{d}_{z,j,{\rm net}}$.
\end{definition}

\begin{lemma}
\label{lem: tight_interval_proof}
Given a schedulable real-time composite task set $\mathcal{T}$ under the $h\text{-}1\text{-}1$ model, if EDF fails to find a feasible network schedule, there exists exactly one actuating segment scheduled in the continuous interval but not effectively included in the interval, where the continuous interval is an ENTI. 
\end{lemma}

\begin{proof}

Based on proof by contradiction, we assume that one of the following cases hold. 
\begin{itemize}[leftmargin=*]
\item \textbf{Case 1}: the network segments scheduled in the continuous interval are all effectively included in the interval.
\item \textbf{Case 2}: there exists exactly one sensing segment scheduled in the continuous interval but not effectively included in the interval.
\item \textbf{Case 3}: there exist at least two network segments scheduled in the continuous interval but not effectively included in the interval.
\end{itemize}

Suppose that EDF generates the network schedule $\mathcal{S}_{{\rm net}}$. Consider that EDF fails to find a feasible network schedule. Let $\tau_{j,{\rm net}}$ be the first network segment scheduled by EDF that misses its deadline $\thickbar{d}_{z,j,{\rm net}}$. There exists the continuous interval $[b_0, \thickbar{d}_{z,j,{\rm net}}]$ which is fully utilized by the network segments. There exists no network segment scheduled in the interval $[b_0-1, b_0]$ or the network segment scheduled in it has a deadline later than $\thickbar{d}_{z,j,{\rm net}}$. 

\vspace{0.05in}
{\bf If case 1 holds}, all the network segments scheduled in $[b_0, \thickbar{d}_{z,j,{\rm net}}]$ are effectively included in the interval, with the network segment $\tau_{z,j,{\rm net}}$ considered, the interval $[b_0, \thickbar{d}_{z,j,{\rm net}}]$ is an ENOI, which is a contradiction to that $\mathcal{T}$ is schedulable according to Lemma \ref{lem:schedulability_m_1_1}.

\vspace{0.05in}
{\bf If case 2 holds}, since there exists exactly one sensing segment scheduled in the continuous interval but not effectively included in the interval, this sensing segment should be scheduled in $[b_0-1, b_0]$ according to EDF. This contradicts to the definition of continuous interval.

\vspace{0.05in}
{\bf If case 3 holds}, there exist at least two network segments whose corresponding effective release times are smaller than $b_0$ among the network segments scheduled in $[b_0, \thickbar{d}_{z,j,{\rm net}}]$. If at least one of the two network segments is a sensing segment, the interval $[b_0-1, b_0]$ must be utilized by this sensing segment, which is a contradiction. Therefore, we consider the case that both of them are actuating segments. Let $\tau_{k,x,a}$ and $\tau_{g,y,a}$ be the two actuating segments satisfiying $\thickbar{r}_{k,x,a}, \thickbar{r}_{g,y,a} < b_0$. For their corresponding sensing segments $\tau_{k,x,s}$ and $\tau_{g,y,s}$, it must hold that $\thickbar{r}_{k,x,s}, \thickbar{r}_{g,y,s} < b_0-1$. Thus, if any of the two sensing segments $\tau_{k,x,s}$ and $\tau_{g,s}$ are scheduled in $[b_0, \thickbar{d}_{z,j,{\rm net}}]$, it is a contradiction that the interval $[b_0-1, b_0]$ is not utilized or by a network segment with the deadline later than $\thickbar{d}_{z,j,{\rm net}}$. Hence, both $\tau_{k,x,s}$ and $\tau_{g,y,s}$ must be scheduled before $b_0-1$. Since the finish times obtained by EDF must be different, at least one of $f_{k,x,s}$ and $f_{g,y,s}$ is smaller than $b_0-1$. Therefore, at least one of the actual release times of $\tau_{k,a}$ and $\tau_{g,y,a}$ is smaller than $b_0$, which allows one of $\tau_{k,x,a}$ and $\tau_{g,a}$ to be scheduled in $[b_0-1, b_0]$. This leads to a contradiction to our assumption. Therefore, there exists exactly one actuating segment whose effective release time is smaller than $b_0$. With the network segment $\tau_{z,j,{\rm net}}$ considered, the interval $[b_0, \thickbar{d}_{z,j,{\rm net}}]$ is an ENTI. We thus proved that there exists exactly one actuating segment which is scheduled in the ENTI $[b_0, \thickbar{d}_{z,j,{\rm net}}]$ but not effectively included in the interval if EDF fails.
\end{proof}

Lemma~\ref{lem: tight_interval_proof} indicates that as long as the actuating segments which do not belong to a given ENTI are prevented from being scheduled inside that interval, EDF can construct a feasible network schedule if it exists. 

Based on Lemma~\ref{lem: tight_interval_proof}, Alg.~\ref{alg1} gives an overview of the CRS algorithm under the $h\text{-}1\text{-}1$ model. The algorithm first constructs the effective timing parameters for all the segments. Based on these effective release times and deadlines, the algorithm identifies every ENTI and modifies the effective timing parameters accordingly. To ensure that the search of the ENTIs is complete, we traverse the effective release times of all sensing segments in the descending order, and the effective deadlines of all actuating segments in the ascending order to construct all candidate intervals. If an ENTI $[t_0, t_1]$ is found, we check all the tasks to modify their deadlines based on the following rule: if the deadline is included in $[t_0, t_1]$ and 
the effective release time of its actuating segment is smaller than $t_0$, we set its deadline to $t_0$. Additionally, the algorithm identifies ENOIs and returns a None value to indicate that this task set is unscheduable. This procedure (line 3-16 in Alg.~\ref{alg1}) repeats until all the ENTIs are identified, which has a time complexity of $O(N^3)$, where $N=\sum_{i=1}^{n} \mathcal{H}/T_i$ is the total number of instances obtained from all the tasks. After that, we employ EDF to construct a feasible schedule. The overall time complexity of Alg.~\ref{alg1} is $O(N^3)$.

\begin{algorithm}[!t]
\small
    \DontPrintSemicolon
    \SetKwInOut{Input}{Input}
    \SetKwInOut{Output}{Output}
    \Input{A real-time composite task set $\mathcal{T}=\{\tau_i\}_{i=1}^n$}
    \Output{A network schedule $\mathcal{S}_{\rm net}$ and a computing schedule $\mathcal{S}_{\rm com}$,  if exist}
    \vspace{0.05in}

    Compute the hyper-period $\mathcal H$ and build the set $I$ of instances of $\mathcal{T}$\;
    Construct the set $O$ of ENOIs and ENTIs\;
    \While{$O \neq \emptyset $}
    {
        
        \If{$\alpha \in O$ is an ENOI}
        { 
            \KwRet None
            \tcp{Algorithm reports a failed case} 
        }
        \If{$\lambda \in O$ is an ENTI $[t_0, t_1]$ }
        { 
            \For{$\tau_{i,j} \in I$}
            {

                \If{$d_{i,j} \in \lambda$ and $\thickbar{r}_{i,j,a} <  t_0$}
                {
                    $d_{i,j} = t_0$\;
                }
                
            }
            $O = O - \{\lambda\}$
        }
        Update the set $O$ based on the modified timing parameters\;
        
    }
    
    Use EDF to schedule the network and computing segments in parallel based on their effective deadlines to obtain the schedules $\mathcal{S}_{{\rm net}}$, and $\mathcal{S}_{{\rm com}}$\;

    \KwRet $\mathcal{S}_{{\rm net}}$, $\mathcal{S}_{{\rm com}}$\;

    \caption{CRS Algorithm ($h\text{-}1\text{-}1$ Model)}
    \label{alg1}
\end{algorithm}

The following theorem proves the correctness of the CRS algorithm under the $h\text{-}1\text{-}1$ model.
\begin{theorem}
If there exists a feasible schedule for the real-time composite task set under the $h\text{-}1\text{-}1$ model, then Alg.~\ref{alg1} can find it. 
\end{theorem}

\begin{proof}
The proposed CRS algorithm under the $h\text{-}1\text{-}1$ model modifies the task deadlines  so that there exist no actuating segments which would be scheduled in an ENTI if they are not effectively included in the interval. 
According to Lemma \ref{lem: tight_interval_proof}, if the single actuating segment which is not included in the ENTI is removed, the EDF is able to find a feasible network schedule $\mathcal{S}_{{\rm net}}$ for the network segments.
Besides, according to Lemma \ref{lem:3_stage_m_1_1}, the computing schedule $\mathcal{S}_{{\rm com}}$ must be feasible as well. 
\end{proof}


\section{CRS Problem under $1\text{-}M\text{-}1$ Model}\label{sec:CRS-1m1}

We now extend the study to the  $1\text{-}m\text{-}1$ model, where every real-time composite task has unit-size execution time for its sensing/actuating segments and arbitrary execution time larger than 1 for its computing segment. This model represents a wide range of NCSs which employ online optimization methods so that their computing time is longer than the transmission time for sensing and actuating.

Under the $1\text{-}m\text{-}1$ model, we face the same difficulty as in the $h\text{-}1\text{-}1$ model. That is, we need to guarantee that the sensing and actuating segments will not be scheduled in an interval with the network resource utilization larger than 100\%. However, different from the $h\text{-}1\text{-}1$ model, infeasible schedule under the $1\text{-}m\text{-}1$ model can also be generated by scheduling the computing segments improperly. Thus, it is important to take into account the scheduling of computing segments in the algorithm design as well to solve the CRS problem.   

In this section, we present an exponential-time optimal algorithm to solve the CRS problem under the $1\text{-}m\text{-}1$ model based on a novel backtracking strategy. The effectiveness of the proposed algorithm has been validated through extensive experimental results (see Section~\ref{sec:evaluation}).

\subsection{Algorithm Overview}

Alg.~\ref{alg2} gives an overview of the CRS algorithm under the $1\text{-}m\text{-}1$ model, which utilizes an iterative two-stage decomposition method. We decompose the CRS problem into two sub-problems including the computing scheduling sub-problem and the network scheduling sub-problem. The variables of the original CRS problem are also divided into a subset of computing variables and a subset of network variables. The first-stage sub-problem is solved over the computing variables. The values of the network variables are determined in the second-stage sub-problems based on the given first-stage solution. If the subsequent sub-problem determines that the previous stage' decisions lead to infeasible schedules, then new constraint(s) will be added to the original CRS problem, which is re-solved until no new constraints can be added.  Failure will be reported if no feasible composite schedule can be constructed.

Based on the algorithm framework above, it is important to guarantee that the new constraints added in each iteration will not jeopardize the schedulability of the original CRS problem. To tackle this challenge, we design a constraint generator based on a novel backtracking strategy to add new constraints to the CRS problem. This is achieved by modifying the timing parameters of the tasks in the interval(s) with the network resource utilization larger than 100\% based on the iterative two-stage decomposition method.

\subsection{Design details of the CRS algorithm under $1\text{-}m\text{-}1$ model}


\subsubsection{Decomposition Method}
We now reformulate the original CRS problem as a two-stage scheduling problem.

\vspace{0.05in}
\noindent{\bf Computing Scheduling Sub-Problem:} Consider a real-time composite task set $\mathcal{T}$ and the hyper-period $\mathcal{H}$. The objective of the computing scheduling sub-problem is to find a feasible computing schedule $\mathcal{S}_{com}$ with a length of $\mathcal{H}$ if it exists so that the effective deadlines of all computing segments are met. 

\eat{
\begin{equation} \label{eq:constr3}
\setlength\abovedisplayskip{5pt}
\setlength\belowdisplayskip{5pt}
\begin{split}
&f_{i,j,c,1}\geq r_{i,j} + 2   \\ 
&f_{i,j,c, C_{i,c}}\leq d_i-C_{i,a} \\
&f_{i,j,c,k+1}\geq f_{i,j,c,k}+1, \quad   \forall k \in \mathbb{Z}\cap [1, C_{i,c}-1]\\
&\forall i\neq x,  \quad \forall p \in \mathbb{Z}\cap[1,C_{i,c}],  \quad \forall q \in \mathbb{Z}\cap[1,C_{j,c}]\\
& f_{i,j,c,p}\leq f_{x,y,c,q}-1 \quad  \text{or} \quad  f_{i,j,c,p}\geq f_{x,y,c,q}+1 
\end{split}
\end{equation}
where the computing variable is bounded by its earliest possible release time $r_{i,j} + 2$ and latest possible deadline $d_{i,j}-C_{i,a}$. After solving the computing scheduling sub-problem, we formulate the network scheduling sub-problem as follows. 
}

\vspace{0.05in}
\noindent{\bf Network Scheduling Sub-Problem:} Given a real-time composite task set $\mathcal{T}$, the hyper-period $\mathcal{H}$ and the computing schedule $\mathcal{S}_{com}$, the objective of the computing scheduling sub-problem is to find a feasible network schedule $\mathcal{S}_{net}$ with a length of $\mathcal{H}$ if it exists so that the network segments can meet the deadlines obtained based on the start times and finish times of the computing segments in $\mathcal{S}_{com}$.

\eat{
\vspace{0.025in}
\begin{equation} \label{eq:constr4}
\begin{split}
&f_{i,j,s,1}\geq r_{i,j} , \quad f_{i,j,a,1}\geq \hat{f}_{i,j,c, C_{i,c}}+1\\
&f_{i,j,s,1} \leq \hat{f}_{i,j,c,1} + 1, \quad f_{i,j,a,1}\leq d_{i,j}\\
&\forall i\neq u, \quad  \forall x,y \in\{a,s\} \\
&f_{i,j,x,1}\leq f_{u,v,y,1}-1 \quad  \text{or} \quad  f_{i,j,x,1}\geq f_{u,v,y,1}+1 
\end{split}
\end{equation}
\vspace{0.025in}

\noindent where $\hat{f}_{i,j,c,1}$ and $\hat{f}_{i,j,c, C_{i,c}}$ are obtained from any feasible solution of the computing scheduling sub-problem. 
\vspace{0.05in}
}

Since each sub-problem above can be taken as a uni-processor scheduling problem, it is intuitive to design a two-stage EDF to solve the CRS problem by solving the computing scheduling sub-problem in the first stage and the network scheduling sub-problem in the second stage. However, the two-stage EDF may not find a feasible schedule for the second-stage problem as it always employs a fixed first-stage solution. Therefore, we utilize a constraint generator to add new constraints to the original CRS problem whenever two-stage EDF finds an infeasible schedule in the second-stage problem.

\setlength{\textfloatsep}{10pt}
\begin{algorithm}[htb]
\small
    \DontPrintSemicolon
    \SetKwInOut{Input}{Input}
    \SetKwInOut{Output}{Output}
    \Input{A real-time composite task set $\mathcal{T}=\{\tau_i\}_{i=1}^n$}
    \Output{A network schedule $\mathcal{S}_{\rm net}$ and a computing schedule $\mathcal{S}_{\rm com}$,  if exist}
    \vspace{0.05in}
    Compute the hyper-period $\mathcal H$ and construct the set $I$ of instances of  $\mathcal{T}$\;
    \While{True}{
        
        $\mathcal{S}_{{\rm com}}$ = \textbf{ComputingScheduling}$(I)$\;
        \If{$\mathcal{S}_{{\rm com}}$ is infeasible}{
            \KwRet  None      
        }
        $\mathcal{S}_{{\rm net}}$ = \textbf{NetworkScheduling}$(I, \mathcal{S}_{{\rm com}})$\;
        \uIf{$\mathcal{S}_{{\rm net}}$ is infeasible}{
            \If{\textbf{ConstraintGenerator}($I$) = None}
            {
                \KwRet  None     
            }
                  
        }
        \Else{
            \textbf{break}
        }
  
    }
    \KwRet $\mathcal{S}_{{\rm net}}$, $\mathcal{S}_{{\rm com}}$\;
    \caption{CRS Algorithm ($1\text{-}m\text{-}1$ Model)}
    \label{alg2}
\end{algorithm}

\subsubsection{Constraint Generator}

we first introduce some definitions and preliminaries to help understand the design of the constraint generator.

\begin{definition}
\label{def:as_EDF}
{\bf Virtual Release Time/Deadline:} Given a set of real-time composite tasks $\mathcal{T} = \{\tau_1, \tau_2,..., \tau_n\}$, for each instance $\tau_{i,j}$ with the release time $r_{i,j}$ and deadline $d_{i,j}$, the virtual release time and virtual deadline of the sensing segment $\tau_{i,j,s}$ is set to be $r_{i,j,s} = r_{i,j}$ and $d_{i,j,s} = s_{i,j,c}$. The virtual release time and deadline of the actuating segment $\tau_{i,j,a}$ is set to be $r_{i,j,a} = f_{i,j,c}$ and $d_{i,j,a} = d_{i,j}$, where $s_{i,j,c}$ and $f_{i,j,c}$ are the start time and the finish time of the corresponding computing segment $\tau_{i,j,c}$, respectively. Both $s_{i,j,c}$ and $f_{i,j,c}$ are obtained in the computing scheduling sub-problem.
\end{definition}

We say a network segment is \textit{virtually included} in a time interval if its virtual release time and virtual deadline are both included in that interval. That is, a network segment $\tau_{i,j,{\rm net}}$ is virtually included in $[t_0, t_1]$ if $r_{i,j,{\rm net}}\geq t_0$ and $d_{i,j,{\rm net}}\leq t_1$. Based on this condition, we define a \textit{virtual network demand} over a given interval to be the sum of the execution time of all the network segments that are virtually included in that interval. We define the \textit{virtual network overload interval} and \textit{minimal virtual network overload interval} as follows.

\begin{definition}
\label{def:overload}
{\bf Virtual Network Overload Interval (VNOI):} Given a set of network segments and a time interval $[t_0, t_1]$, the interval $[t_0, t_1]$ is a VNOI if its virtual network demand is larger than $t_1 - t_0$.    
\end{definition}

\begin{definition}
\label{def:Minium}
{\bf Min. Virtual Network Overload Interval:} Given a VNOI $[t_0, t_1]$, if there does not exist another VNOI $[t_2, t_3]$ satisfying $[t_2, t_3]\subset[t_0, t_1]$, $[t_0, t_1]$ is a minimal VNOI.
\end{definition}

\begin{figure}
  \centering
  \includegraphics[width=0.9\columnwidth]{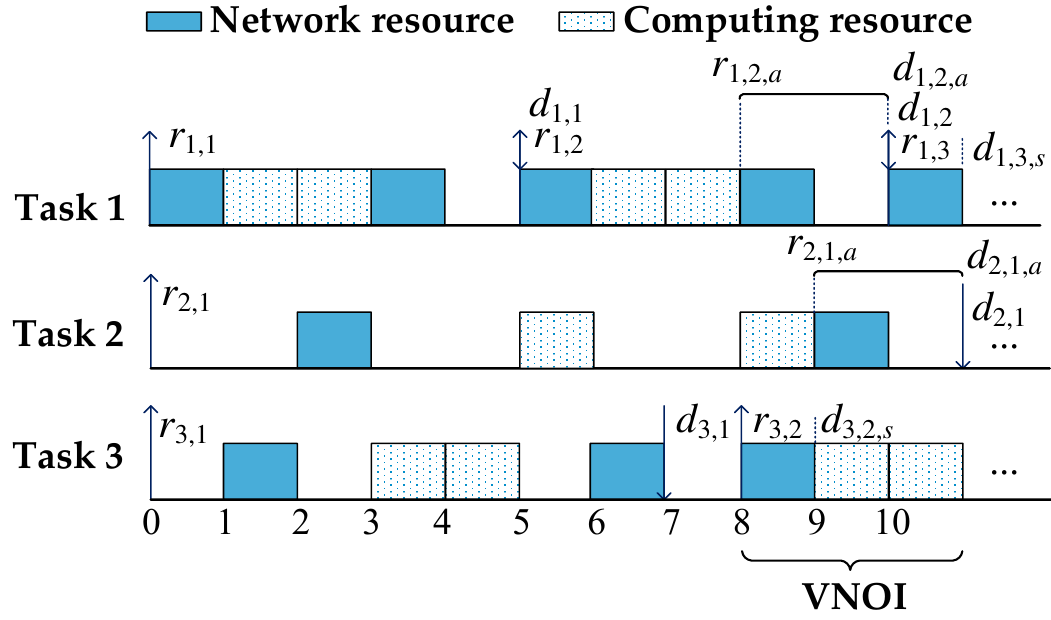}
  \caption{An instance of VNOI}
  \label{fig:eg2}
\end{figure}

\begin{example}
Consider three tasks $\tau_1 = (5, 4, 1, 2, 1)$, $\tau_2 = (20, 11, 1, 2, 1)$, and $\tau_3 = (8, 7, 1, 2, 1)$. Suppose EDF has scheduled the computing segments. As shown in Fig.~\ref{fig:eg2}, based on the virtual timing parameters, the network segments $\tau_{1,2,a}$, $\tau_{2,1,a}$, $\tau_{1,3,s}$, and $\tau_{3,2,s}$ satisfy that $r_{1,2,a} = 8$,  $d_{1,2,a} = 10$, $r_{2,1,a} = 9$, $d_{2,1,a} = 11$, $r_{1,3,s} = 10$,   $d_{1,3,s} = 11$, $r_{3,2,s} = 8$, and  $d_{3,2,s} = 9$. Thus, the virtual network demand over $[8, 11]$ is $4$ and it is an VNOI.   
\end{example}

For a VNOI $\alpha=[t_0, t_1]$, we use $U_\alpha$ to denote the set of network segments which are virtually included in $\alpha$. We use $A_{\alpha} \subset U_\alpha$ to denote the set of actuating segments in $U_\alpha$ whose corresponding sensing segments are not virtually included in $\alpha$. That is, for any $\tau_{i,a} \in A_{\alpha}$, $\tau_{i,a} \in U_{\alpha}$ and $\tau_{i,s} \notin U_{\alpha}$. We use $S_{\alpha} \subset U_{\alpha}$ to denote the set of sensing segments in $U_\alpha$ whose corresponding actuating segments are not virtually included in $\alpha$. Finally, we use $M_{\alpha} \subset U_{\alpha}$ to denote the set of networking segments which contain both the sensing and actuating segments from the same tasks.

According to the definitions of $U_{\alpha}$, $A_{\alpha}$, $S_{\alpha}$, and $M_{\alpha}$, it holds that  $U_{\alpha} = A_{\alpha} \cup S_{\alpha} \cup M_{\alpha}$. Let $A_{\alpha, c}$, $S_{\alpha, c}$ and $M_{\alpha, c}$ denote the sets of corresponding computing segments of the network segments in $A_{\alpha}$, $S_{\alpha}$ and $M_{\alpha}$, respectively. We define $\phi(A_{\alpha, c})$, $\phi(S_{\alpha, c})$, and $\phi(M_{\alpha, c})$ to be the amount of execution time of the computing segments in $A_{\alpha, c}$, $S_{\alpha, c}$, and $M_{\alpha, c}$ that are scheduled in the interval $[t_0-1, t_1+1]$ in the computing schedule.

It should be noted that there exists no feasible schedule for the network scheduling sub-problem when VNOIs are present. Thus, the constraint generator must eliminate all the VNOIs by adding new constraints to the original CRS problem. In the following, we first present several important observations and then show how to eliminate a VNOI by adjusting the timing parameters of the network segments that are virtually included in that interval. 

\begin{lemma}
\label{lem:1}
The virtual network demand over a VNOI $\alpha =  [t_0, t_1]$ is at most $t_1 - t_0+2$ under the $1\text{-}m\text{-}1$ model. 
\end{lemma}

\begin{proof}

Suppose the interval $\alpha = [t_0, t_1]$ has a total amount of $t_1 - t_0 + x$ units of network segments virtually included in $[t_0, t_1]$, where $x\in \mathbb{Z}_{>0}$. We will prove $x\leq 2$ as follows. 

Based on Definition~\ref{def:as_EDF}, the virtual release times of every two actuating segments are different since the finish times of their corresponding computing segments cannot be identical in the computing schedule. Thus, for any two different actuating segments $\tau_{i,j,a},\tau_{x,y,a}\in U_{\alpha}$, it holds that $|r_{i,j,a}-r_{x,y,a}| \geq 1$. Similarly, the virtual deadlines of every two sensing segments are different since the start times of their corresponding computing segments cannot be identical in the computing schedule. Thus, for any two different sensing segments $\tau_{i,j,s},\tau_{x,y,s} \in U_{\alpha}$,  the condition $|d_{i,j,s} - d_{x,y,s}| \geq  1$ holds.

For any $\tau_{i,j,a} \in A_{\alpha}$, since its corresponding sensing segment is excluded from $U_{\alpha}$, its corresponding computing segment $\tau_{i,j,c} $ has the release time $r_{i,j,c}\leq t_0$ and deadline $d_{i,j,c}\leq t_1-1$.  The last unit of $\tau_{i,j,c}$ is scheduled in $[r_{i,j,a}-1, r_{i,j,a}]$ in the computing schedule. Due to that $r_{i,j,a}\in [t_0, t_1-1]$, the last unit of $\tau_{i,j,c}$ can be scheduled in $[t_0-1, t_1-1]$.  Taking into account all the computing segments in $A_{\alpha, c}$, it holds that $|A_{\alpha}|\leq \phi(A_{\alpha,c})$. Similarly, for any $\tau_{x,y,s} \in S_{\alpha}$, since its corresponding actuating segment is excluded from $U_{\alpha}$, its corresponding computing $\tau_{x,y,c}$ has release time $r_{x,y,c}> t_0$ and deadline $d_{x,y,c}\geq t_1+1$. The first unit of $\tau_{x,y,c}$ is scheduled in $[d_{x,y,s}-1,d_{x,y,s}]$ in the computing schedule. Due to that $d_{x,y,s} \in [t_0+1, t_1]$, at least the first unit of $\tau_{i,j,c}$ is scheduled in $[t_0+1, t_1+1]$. Considering all the computing segments in $S_{\alpha}$, it holds that $|S_{\alpha}|\leq \phi(S_{\alpha,c})$. For any actuating segment $\tau_{u,v,a}$ in $M_{\alpha}$, since its corresponding sensing segment and itself are both in the interval $[t_0, t_1]$, it holds that  $r_{u,v,c}\geq t_0+1$ and $d_{u,v,c}\leq t_1-1$. So the computing segment $\tau_{u,v,c}$ is scheduled within $[t_0+1, t_1-1]$. Since $\tau_{u,v,c}$ has at least two units of execution time, it holds that  $|M_{\alpha}|\leq \phi(M_{\alpha,c})$. Based on the above constraints on the timing parameters of the computing segments, it holds that
\begin{equation} \label{eq:6}
|S_{\alpha}|+|A_{\alpha}|+|M_{\alpha}|\leq \phi(S_{\alpha,c})+\phi(A_{\alpha,c})+\phi(M_{\alpha,c})
\end{equation}

As the amount of the execution time of the computing segments scheduled in the interval $[t_0-1, t_1+1]$  cannot exceed the length $t_1-t_0+2$ of this interval, there is 
\begin{equation} \label{eq:7}
\phi(S_{\alpha,c})+\phi(A_{\alpha,c})+\phi(M_{\alpha,c}) \leq t_1-t_0+2
\end{equation}

Combining \eqref{eq:6} and \eqref{eq:7}, it yields that
\begin{equation} \label{eq:8}
|S_{\alpha}|+|A_{\alpha}|+|M_{\alpha}|\leq t_1-t_0+2
\end{equation}

Therefore, we conclude that $x\leq 2$ and the virtual network demand over an interval $[t_0, t_1]$ is at most $t_1 - t_0+2$ if the EDF schedule of the computing segments is feasible. 
\end{proof}

Lemma \ref{lem:1} gives an upper bound on the virtual network demand over a VNOI. The number of overflow network segments in the interval is thus restricted to 2. However, selecting two feasible overflow network segments from a VNOI is still challenging due to its combinatorial nature. The following lemma presents an important property of the minimal VNOI (see Definition~\ref{def:Minium}), which can further speed up the overflow segment selection from VNOIs.

\begin{lemma}
\label{lem:2}
The virtual network demand over a minimal VNOI $\alpha =  [t_0, t_1]$ is at most $t_1 - t_0+1$ under the $1\text{-}m\text{-}1$ model. 
\end{lemma}

\begin{proof}

According to Lemma \ref{lem:1}, we have either 
\begin{equation}\label{eq:9}
|A_{\alpha}|+|M_{\alpha}|+|S_{\alpha}|= t_1-t_0+1
\end{equation}
or
\begin{equation}\label{eq:10}
|A_{\alpha}|+|M_{\alpha}|+|S_{\alpha}|= t_1-t_0+2
\end{equation}
for the virtual network demand over the minimal VNOI $[t_0, t_1]$.  

We first assume that the virtual network demand is $t_1 - t_0+2$ and prove that this leads to a contradiction and thus the demand can only be $t_1 - t_0+1$.  According to \eqref{eq:6} and \eqref{eq:7}, one can derive that
\begin{equation} \label{eq:11}
\phi(A_{\alpha,c})+\phi(M_{\alpha,c})+\phi(S_{\alpha,c}) = t_1-t_0+2
\end{equation}

This indicates that the interval $[t_0-1, t_1+1]$ is fully utilized by the computing segments in $A_{\alpha,c}$, $M_{\alpha,c}$ and $S_{\alpha,c}$. Also, because of \eqref{eq:10}, \eqref{eq:11}, and the following inequalities
\begin{equation}\label{eq:12}
\begin{split}
|A_{\alpha}| \leq \phi(A_{\alpha,c}); |S_{\alpha}| \leq \phi(S_{\alpha,c}); |M_{\alpha}| \leq \phi(M_{\alpha,c})
\end{split}
\end{equation}
one can obtain that 
\begin{equation}\label{eq:13}
\begin{split}
|A_{\alpha}| = \phi(A_{\alpha,c});
|S_{\alpha}| = \phi(S_{\alpha,c});
|M_{\alpha}| = \phi(M_{\alpha,c})
\end{split}
\end{equation}
where $|A_{\alpha}|\geq 2$, $|S_{\alpha}|\geq 2$, and $|M_{\alpha}|\geq 0$ hold based on the following discussion. Suppose that $|A_{\alpha}|=1$, it then holds that $\phi(A_{\alpha,c})=1$. As the computing segments in $S_{\alpha,c}\cup M_{\alpha,c}$ can only be scheduled in $[t_0+1, t_1+1]$, it holds that $\phi(S_{\alpha,c})+\phi(M_{\alpha,c})\leq t_1-t_0$ and thus $|S_{\alpha}|+|M_{\alpha}|\leq t_1-t_0$ according to \eqref{eq:13}. As $|A_{\alpha}|=1$, the number of network segments $|A_{\alpha}|+|M_{\alpha}|+|S_{\alpha}|$ is at most $t_1-t_0+1$, which contradicts our assumption \eqref{eq:10}. Similarly, $|S_{\alpha}|\geq 2$ holds.

Let $\tau_{u,v,a}\in A_{\alpha} \cup M_{\alpha}$ be the actuating segment that has the latest release time among all the actuating segments, where it holds that $r_{u,v,a}\geq t_0$. Since $|A_{\alpha}|\geq 2$, it meets that $r_{u,v,a}>t_0$. Let $S_{\alpha,c,1} \subset S_{\alpha,c}$ be the set of computing segments scheduled before $r_{u,v,a}$ and $S_{\alpha,c,2}\subset S_{\alpha,c}$ be the set of computing segments scheduled after $r_{u,v,a}$.  The computing segments in $A_{\alpha,c} \cup M_{\alpha,c,} \cup S_{\alpha,c,1}$ are finished before $r_{u,v,a}$ while the computing segments in $S_{\alpha,c,2}$ fully utilize the interval $[r_{u,v,a}, t_1+1]$ based on the condition \eqref{eq:10}. 

With the condition $|S_{\alpha}|=\phi(S_{\alpha,c})$, for any $\tau_{i,j,c} \in S_{\alpha,c}$, it only has its first unit scheduled in the interval $[t_0+1, t_1+1]$ in the computing schedule. We define that $S_{\alpha,c,2}= \{\tau_{g+1,z,c},\tau_{g+2,z,c},...,\tau_{g+|S_{\alpha,c,2}|,z,c}\}$. For any $j \in [1, |S_{\alpha,c,2}|]$, we assume that $\tau_{g+j,z,c}$ has its first unit scheduled in the interval $[r_{u,v,a}+j-1, r_{u,v,a}+j]$. Since $\tau_{g+j,z,c}$ takes at least two units of execution time, $\tau_{g+j+1,z,c}$ must preempt $\tau_{g+j,z,c}$ at time $r_{u,v,a}+j$ so that $\tau_{g+j+1,c}$ could have its first unit scheduled in the interval $[r_{u,v,a}+j, r_{u,v,a}+j+1]$. Therefore, it holds that $r_{g+j+1,z,c}=r_{u,v,a}+j$ and $d_{g+j+1,z,c} < d_{g+j,z,c}$ for any $j\in [1,|S_{\alpha,c,2}|-1]$ if $|S_{\alpha,c,2}|\geq 2$. The corresponding sensing segment $\tau_{g+j+1,z,s}$ satisfies that $r_{g+j+1,z,s} = r_{u,v,a}+j-1$ and $d_{g+j+1,z,s} = r_{u,v,s}+j$. Since the computing segments in $S_{\alpha,c,2}$ scheduled in $[r_{u,v,a}+1, t_1+1]$ are preempted by one another, there are $t_1-r_{u,v,a}$ sensing segments virtually included in $[r_{u,v,a}, t_1]$. Similarly, there are $d_{u,v,a}-r_{u,v,a}$ sensing segments virtually included in the interval $[r_{u,v,a}, d_{u,v,a}]$. With $\tau_{u,v,a}$ considered, the total number of network segments virtually included in the interval $[r_{u,v,a}, d_{u,v,a}]$ is $d_{u,v,a}-r_{u,v,a}+1$. This indicates that the interval $[r_{u,v,a}, d_{u,v,a}]\subset [t_0, t_1]$ is also an VNOI, which contradicts that the interval $[t_0, t_1]$ is already a minimal one.  Therefore, the virtual network demand over the minimal VNOI $[t_0, t_1]$ can only be $t_1 - t_0+1$, i.e., $|A_{\alpha}|+|M_{\alpha}|+|S_{\alpha}|= t_1-t_0+1$.
\end{proof}

Lemma~\ref{lem:2} indicates that only one network segment needs to be moved out of a minimal VNOI. Based on this observation, we aim to select a network segment as the overflow segment and modify its timing parameters to prevent it from being scheduled in the interval. This procedure is  defined as the elimination procedure. We first identify the {\em precondition} for a network segment to be selected as an overflow segment. That is, after the modification of the timing parameter(s) of its corresponding real-time composite task, the schedulability of the network scheduling sub-problem and computing scheduling sub-problem can still be preserved. Based on this precondition, we design the constraint generator based on a backtracking strategy to eliminate all the VNOIs.

The elimination procedure for a minimal VNOI $\alpha = [t_0, t_1]$ is achieved based on a selected candidate segment $\tau_{i,j,{\rm net}}$ and the instance set $I$. Consider that we select an actuating segment $\tau_{i,j,a} \in A_{\alpha}$. Let $D$ be the set of timing parameters which include the deadlines of the actuating segments in $A_{\alpha}$ and $t_0$. After sorting $D$ in the descending order, we can find $d_0\in D$, the latest element smaller than $\thickbar{d}_{i,j,a}$. We modify deadline $\thickbar{d}_{i,j,a}$ to $d_{0}$ and $\thickbar{d}_{i,j,c}$ to $d_0-1$. This assures that the priority of $\tau_{i,j,c}$ improves, thus guaranteeing that their corresponding actuating segments will not be scheduled in $\alpha$ to make the interval become virtually overload again. After rescheduling the computing segments, the minimal VNOI $\alpha$ is eliminated.   

On the other hand, if we select a sensing segment $\tau_{i,j,s}\in S_{\alpha}$, we obtain the set $R$ of timing parameters including the release times of the sensing segments in $S_{\alpha}$ and $t_1$, and sort them in the ascending order. Let $r_0\in R$ be the first element larger than $\thickbar{r}_{i,j,s}$, and we modify $\thickbar{r}_{i,j,s}$ to $r_{0}$. The minimal VNOI can be eliminated by rescheduling the computing segments and updating the virtual timing parameters.

Note that although any network segment virtually included in the interval $[t_0, t_1]$ can be taken as a candidate for the overflow segment, a candidate is infeasible if modifying its timing parameter hurts the schedulability of the real-time composite task set. Therefore, we design the constraint generator to eliminate the VNOIs based on a backtracking algorithm.

 \begin{algorithm}[!t]
 \small
 \DontPrintSemicolon
    \SetKwInOut{Input}{Input}
    \SetKwInOut{Output}{Output}
    \Input{An instance set $I_0$ of real-time tasks }
    \Output{A modified instance set $I_2$ or the initial instance set $I_0$}
    \vspace{0.05in}
    
    \If{the iterative two-stage decomposition method goes to network scheduling}
    {
    Construct the set $O$ of minimal VNOIs\;
    \If{$O = \emptyset$}
    {
          \KwRet $I_0$\;
    }
    
    Get the earliest minimal VNOI $\alpha \in O$\;

    \For {$\tau_{i,{\rm net}} \in A_{\alpha}\cup S_{\alpha} $}
    {
        $I_1$ = Eliminate$(I_0, \alpha, \tau_{i,{\rm net}})$\;
        $I_2 = \textbf{\textup{ConstraintGenerator}}(I_1)$\;
        \If{$I_2 \neq$ None}
        {
            \KwRet $I_2$\;
        }

    }
    }
    \KwRet None \tcp{Algorithm reports a failed case}
    \caption{Constraint Generator ($1\text{-}m\text{-}1$ Model)}
    \label{alg3}
\end{algorithm}

Alg.~\ref{alg3} shows an overview of the constraint generator. The key idea is to utilize the backtracking search to eliminate all the VNOIs through the constraint generator, where the input is the initial set of instances and the output is either a feasible instance set or reported failure. The constraint generator works as a recursion tree. The root corresponds to the original problem of finding a feasible instance set from the given instance set. Each node in this tree corresponds to a recursive sub-problem. In particular, when leaves cannot be further extended, it is either because the network scheduling sub-problem based on the current instance set finds infeasible schedules, or because the feasible instance set is found and returned. 
When the constraint generator is called, we construct the set $O$ of minimal VNOIs based on the current instance set. If $O$ is not empty, then starting with the earliest minimal VNOI in $O$, we branch the search based on its overflow candidate segments. For each candidate segment in the minimal VNOI, we call its elimination procedure to eliminate the overload interval. If every candidate segment fails to utilize the constraint generator to find an instance set which generates no VNOIs in the two-stage decomposition method, it indicates that the task set is not schedulable, which returns a None value and reports failure. The time complexity of the backtracking algorithm is $O(W^pN^2)$, where $W$ is the average number of candidate segments in the minimal overload interval and $p$ is the number of the overload intervals. When $W$ approaches $N$, $p$ goes to $1$. On the contrary, when $W$ approaches $1$, $p$ goes to $N$, where $N=\sum_{i=1}^{n} \mathcal{H}/T_i$ is the total amount of instances of all the tasks. Each node takes $O(N^2)$ time to identify the intervals.

\begin{theorem}
If a feasible schedule exists for the real-time composite task set $\mathcal{T}$ under $1\text{-}m\text{-}1$ model, Alg.~\ref{alg2} can find it.
\end{theorem}
\begin{proof}
Since we employ the backtracking strategy to enumerate all the candidates of network segments to eliminate all the VNOIs, it guarantees to remove all the overload intervals if there exists a feasible elimination process. This shows that the proposed scheduling algorithm can find a feasible schedule if there exists one. 
\end{proof}

\section{CRS Problem under General Model}\label{sec:CRS-1m1}

Based on the studies of the two special cases, we now extend the CRS problem to its general case, where every real-time composite task has arbitrary execution time for each segment.

\subsection{Algorithm Overview}

We first extend the concepts of effective network overload/tight intervals (ENTI/ENOI, see Definition~\ref{def:tig}) to \textit{effective computing overload/tight intervals}.

\begin{definition}
\label{def:ecoi}
{\bf Effective computing overload/tight interval (ECOI/ECTI):} Given a set of real-time composite tasks and a time interval $[t_0, t_1]$, $[t_0, t_1]$ is an effective computing overload interval if the effective computing demand over $[t_0, t_1]$ is larger than $t_1-t_0$. $[t_0, t_1]$ is an effective computing tight interval if the effective computing demand over $[t_0, t_1]$ is equal to $t_1-t_0$. 
\end{definition}

For the CRS problem in the general case, we first use EDF to schedule the network and computing segments in parallel based on their effective deadlines. Suppose EDF fails scheduling a segment and generates a partial schedule, we define the \textit{provisional timing parameters} as follows.  

\begin{definition}
\label{def:pEDF}
{\bf Provisional Release Time/Deadline:} For each instance $\tau_{i,j}$ with release time $r_{i,j}$ and deadline $d_{i,j}$, the provisional release time $\hat{r}_{i,j,s}$ and provisional deadline $\hat{d}_{i,j,s}$ of its sensing segment $\tau_{i,j,s}$ are $\hat{r}_{i,j,s} =r_{i,j}$ and  $\hat{d}_{i,j,s} = f_{i,j,s}$, where $f_{i,j,s}$ is the finish time of the sensing segment. The provisional release time $\hat{r}_{i,j,c}$ and provisional deadline $\hat{d}_{i,j,c}$ of its computing segment $\tau_{i,j,c}$ are $\hat{r}_{i,j,c} = f_{i,j,s}$ and $\hat{d}_{i,j,c} = f_{i,j,c}$, where $f_{i,j,c}$ is the finish time of the computing segment. The provisional release time $\hat{r}_{i,j,c}$ and provisional deadline $\hat{d}_{i,j,a}$ of its actuating segment $\tau_{i,j,a}$ are $\hat{r}_{i,j,a} = f_{i,j,c}$ and $\hat{d}_{i,j,a} = d_{i,j}$.
\end{definition}

The provisional timing parameters of each segment are initialized as its effective timing parameters before EDF is applied. We say a network/computing segment is \textit{provisionally included} in a time interval if its provisional release time and provisional deadline are both within that interval. Based on this condition, we define a \textit{provisional network/computing demand} over a given interval to be the sum of the execution time of all the network/computing segments that are provisionally included in that interval. We define the \textit{provisional network/computing overload interval} as follows.

\begin{definition}
\label{def:pnoi}
{\bf Provisional Network/Computing Overload Interval (PNOI/PCOI):} Given a set of network/computing segments and a time interval $[t_0, t_1]$, $[t_0, t_1]$ is a provisional network/computing overload interval if its provisional network/computing demand is larger than $t_1 - t_0$.    
\end{definition}

\begin{figure}
  \centering
  \includegraphics[width=0.88\columnwidth]{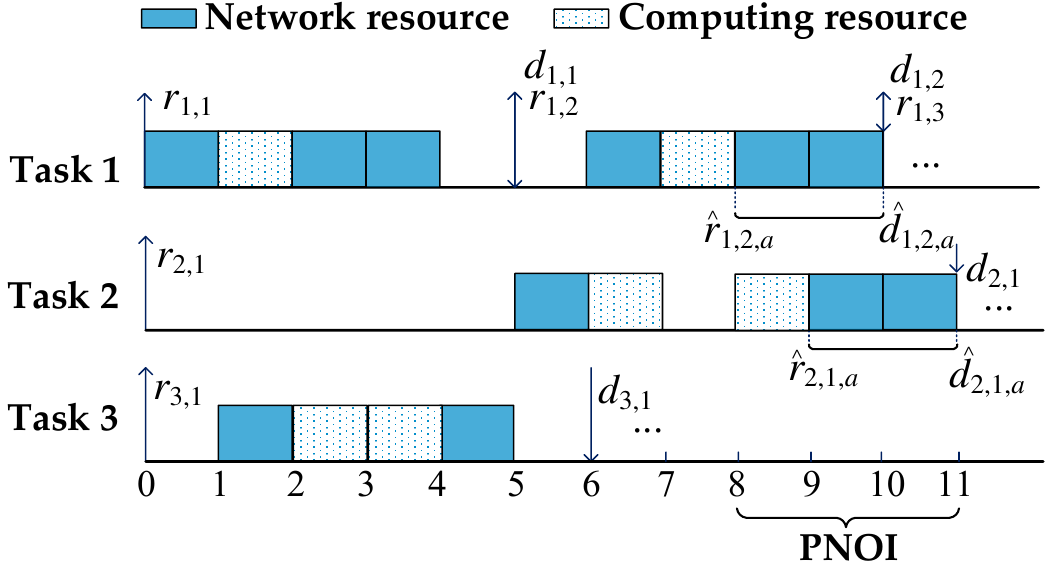}
  \caption{An instance of PNOI}
  \label{fig:eg3}
\end{figure}

\begin{example}
Consider three tasks $\tau_1 = (5, 4, 1, 1, 2)$, $\tau_2 = (20, 11, 1, 2, 2)$, and $\tau_3 = (50, 6, 1, 2, 1)$. Suppose EDF is employed to schedule the network and computing segments and fails at time $11$. As shown in Fig.~\ref{fig:eg3}, based on the provisional timing parameters, the network segments $\tau_{1,2,a}$, $\tau_{2,1,a}$ satisfy that $\hat{r}_{1,2,a} = 8$,  $\hat{d}_{1,2,a} = 10$, $\hat{r}_{2,1,a} = 9$, and $\hat{d}_{2,1,a} = 11$. So the provisional network demand over $[8, 11]$ is $4$, and it is an PNOI.   
\end{example}

For a PNOI $\alpha=[t_0, t_1]$, we use $\mathcal{X}_{\alpha}$ to denote the set of actuating segments provisionally included but not effectively included in $[t_0, t_1]$. Similarly, for a PCOI $\beta=[b_0, b_1]$, we use $\mathcal{C}_\beta$ to denote the set of computing segments that are provisionally included but not effectively included in $[b_0, b_1]$.

Alg.~\ref{alg4} gives an overview of the CRS algorithm under the general model. It is an iterative two-stage method utilizing two groups of intervals to modify the timing parameters of the tasks. \eat{The key idea is to iteratively find the intervals with absolute $100\%$ utilization (Stage 1) and the intervals with the provisional utilization (to be defined later) larger than $100\%$ (Stage 2), and to move out the extra segments appear in those intervals.} 

In stage 1, effective network/computing tight intervals (ENTI/ECTI, see Definition \ref{def:tig} and \ref{def:ecoi}) are used to represent the intervals that have $100\%$ utilization in computing and network resources in any feasible schedule, respectively. A network/computing segment with only one of its effective release time and effective deadline included in an ENTI/ECTI must be moved out of that interval. This adjustment provides better priority for those segments and prevents them from being scheduled in an ENTI/ECTI. After modifying the timing parameters, EDF is used to schedule the network and computing segments in parallel based on their effective deadlines. If EDF fails to find a feasible schedule, we turn to stage 2. 

In stage 2, provisional network/computing overload intervals (PNOI/PCOI, see Definition \ref{def:pnoi}), obtained by EDF scheduling are used to represent the intervals that have the provisional utilization larger than $100\%$. This indicates that we need to move some segments out of a PNOI/PCOI. Among all the candidate segments in a PNOI/PCOI, we propose a greedy strategy to select the actuating/computing segment with the earliest effective release time and move its effective deadline ahead. This will make its corresponding preceding segment(s), e.g., the sensing or/and computing segment, finish earlier, and create more space for scheduling the current segment.

\setlength{\textfloatsep}{10pt}
\begin{algorithm}[!t]
\small
    \DontPrintSemicolon
    \SetKwInOut{Input}{Input}
    \SetKwInOut{Output}{Output}
    \Input{A real-time composite task set $\mathcal{T}=\{\tau_i\}_{i=1}^n$}
    \Output{A network schedule $\mathcal{S}_{\rm net}$ and a computing schedule $\mathcal{S}_{\rm com}$,  if exist}
    \vspace{0.05in}
    Compute the hyper-period $\mathcal H$ and construct the set $I$ of instances of $\mathcal{T}$\;
    \While{True}{
        \tcp{Stage-1}
        \If{an ENOI/ECOI is identified}
        {
            \KwRet  None \tcp{reports a failed case} 
        }
        \If{an identified ENTI/ECTI triggers the modification on timing parameters (Rule 1(a-b) and Rule 2(a-b))}
        {
            \textbf{continue}
        }
        \tcp{Stage-2} 
        Run EDF to schedule the network and computing segments in parallel based on their effective deadlines to construct $\mathcal{S}_{{\rm net}}$ and $\mathcal{S}_{{\rm com}}$. EDF terminates if any segment misses the deadline. The provisional timing parameters are updated\;
       \If{the earliest PCOI $\beta = [b_0, b_1]$ is located}
       {
            
            \If{$\mathcal{C}_\beta$ has no feasible candidate segment}
            {
                \KwRet  None
            }
            Modify the candidate segment in $\mathcal{C}_\beta$ with earliest effective release time (Rule 3(a-b))\;
            \textbf{continue}
       }
        \If{the earliest PNOI $\alpha = [t_0, t_1]$ is located}
       {
            \If { $\mathcal{X}_{\alpha}$ has no feasible candidate segment}
            {
                \KwRet  None
            }

            Modify the actuating segment in $\mathcal{X}_{\alpha}$ with earliest effective release time (Rule 4(a-b))\;
            \textbf{continue}

       }
        \textbf{break}\;

    }
    \KwRet $\mathcal{S}_{{\rm net}}$, $\mathcal{S}_{{\rm com}}$\;
    \caption{CRS Algorithm (General Model)}
    \label{alg4}
\end{algorithm}

\subsection{Design details of the CRS algorithm under general model}

We now present the details of the greedy heuristics. \eat{ based on an iterative two-stage framework summarized in Alg.~\ref{alg4}. In the first stage, we utilize the ENTIs and ECTIs to build the constraints. Then in the second stage, we introduce a greedy modification strategy based on the located PNOIs and PCOIs to modify the effective deadlines of the segments.}

\subsubsection{Stage 1: Modification based on effective timing parameters}

The first stage of the algorithm identifies ECOIs and ENOIs to decide the schedulability of the task set. It then utilizes ENTIs and ECTIs to adjust the tasks according to the following procedure. 

Consider an ENTI $\alpha = [t_0, t_1]$, for any sensing/actuating segment $\tau_{i,j,s}$/$\tau_{i,j,a}$, we introduce Rule 1a and Rule 1b to modify its timing parameters, respectively. 

\begin{itemize}
    \item Rule 1a: if $\tau_{i,j,s}$ is not effectively included in $\alpha$ but has its effective release time $\thickbar{r}_{i,j,s}$ included in $\alpha$, we modify $\thickbar{r}_{i,j,s}$ to $t_1$ and the effective release times of its computing and actuating segments are adjusted to $t_1 + C_{i,s}$ and $t_1 + C_{i,s} + C_{i,c}$, respectively. 
    \item Rule 1b: if $\tau_{i,j,a}$ is not effectively included in $\alpha$ but has its effective deadline $\thickbar{d}_{i,j,a}$ included in $\alpha$, we modify $\thickbar{d}_{i,j,a}$ to $t_0$ and the effective deadlines of its computing and sensing segments are adjusted to $t_0 - C_{i,a}$ and $t_0 - C_{i,a} - C_{i,c}$, respectively. 
\end{itemize}

 Similarly, for any ECTI $\alpha = [t_0, t_1]$, we introduce Rule 2a and Rule 2b to modify the timing parameters of a computing segment $\tau_{i,j,c}$:
 
 \begin{itemize}
     \item Rule 2a: if $\tau_{i,j,c}$ is not effectively included in $\alpha$ but has its effective release time $\thickbar{r}_{i,j,c}$ included in $\alpha$, we modify $\thickbar{r}_{i,j,c}$ to $t_1$ and the release time of its actuating segment is adjusted to $t_1 + C_{i,c}$. 
     \item Rule 2b: if $\tau_{i,j,c}$ is not effectively included in $\alpha$ but has its effective deadline $\thickbar{d}_{i,j,c}$ included in $\alpha$, we modify $\thickbar{d}_{i,j,c}$ to $t_0$ and the effective deadline of its sensing segment is adjusted to $t_0 - C_{i,c}$.
 \end{itemize}

To thoroughly check all the identified ECTIs and ENTIs and the newly generated ones due to the ongoing modifications, whenever a modification is made, we check again on all the ECTIs and ENTIs based on the modified task set until no more modifications is needed. In the meanwhile, we utilize ECOIs and ENOIs as the necessary conditions to decide the feasibility of the task set. If any ECOI or ENOI is found, the task set is unschedulable.

\subsubsection{Stage 2: Modification based on provisional timing parameters}

After stage 1, we use EDF to schedule the segments based on their effective deadlines. If EDF fails to construct a feasible schedule, the algorithm moves to the second stage to deal with PCOIs and PNOIs by employing greedy strategy that always modifies the candidate segment with the earliest effective release time.

Consider the first case that we locate the earliest PCOI $\beta = [b_0, b_1]$, where the provisional demand is $\mathcal{D}_{\beta}$ and the extra provisional demand satisfies $\mathcal{D}_{\beta, extra} = \mathcal{D}_{\beta} - b_1 + b_0$. Let $\mathcal{C}_\beta$ denote the set of computing segments included in $\beta$. We sort them by their effective release times in ascending order. When traversing $\mathcal{C}_\beta$, for a candidate computing segment $\tau_{i,j,c}$, if there exists a set of computing segments $\mathcal{C}_{\beta,1} \subset \mathcal{C}_\beta$  that have their effective deadlines smaller than $\thickbar{d}_{i,j,c}$, we introduce Rule 3a to modify $\tau_{i,j,c}$; if $\mathcal{C}_{\beta,1}$ does not exist, we introduce Rule 3b to modify $\tau_{i,j,c}$. If the modification generates no further ECOIs or ENOIs, we accept this modification and use EDF to schedule the segments again.

\begin{itemize}
    \item Rule 3a: we modify the effective deadline $\thickbar{d}_{i,j,c}$ of $\tau_{i,j,c}$ to the effective deadline $\thickbar{d}_{x,y,c}$ of another computing segment $\tau_{x,y,c}$ in $\mathcal{C}_{\beta,1}$. $\thickbar{d}_{x,y,c}$ is the latest one among all the computing segments in $\mathcal{C}_{\beta,1}$.
    \item Rule 3b: if $\mathcal{D}_{\beta, extra} \geq C_{i,c}$, we modify $\thickbar{d}_{i,j,c}$ to $b_0$. If $\mathcal{D}_{\beta, extra} < C_{i,c}$, we modify $\thickbar{d}_{i,j,c}$ to $b_0 + C_{i,c} - \mathcal{D}_{\beta, extra}$. 
\end{itemize}

We further consider the second case where the earliest PNOI $\alpha = [t_0, t_1]$ is found. Its provisional demand is $\mathcal{D}_{\alpha}$ and the extra provisional demand is $\mathcal{D}_{\alpha, extra} = \mathcal{D}_{\alpha} - t_1 + t_0$. After sorting the set of actuating segments in $\mathcal{X}_\alpha$ by their effective release times in ascending order, we traverse $\mathcal{X}_\alpha$. For a candidate actuating segment $\tau_{i,j,a}$, if there exists a set of actuating segments $\mathcal{X}_{\alpha,1} \subset \mathcal{X}_\alpha$ that have their effective deadlines smaller than $\thickbar{d}_{i,j,a}$, we introduce Rule 4a to modify $\tau_{i,j,a}$; if $\mathcal{X}_{\alpha,1}$ does not exist, we introduce Rule 4b to modify $\tau_{i,j,a}$.

\begin{itemize}
    \item Rule 4a: we modify the effective deadline $\thickbar{d}_{i,j,a}$ to the effective deadline of another actuating segment $\tau_{x,y,a}$ in the set $\mathcal{X}_{\alpha,1}$ which has the latest effective deadline. 
    \item Rule 4b: if $\mathcal{D}_{\alpha, extra} \geq C_{i,a}$, we modify $\thickbar{d}_{i,j,a}$ to $t_0$. If $\mathcal{D}_{\alpha, extra} < C_{i,a}$, we modify $\thickbar{d}_{i,j,a}$ to $t_0 + C_{i,a} - \mathcal{D}_{\alpha, extra}$.  
\end{itemize}

Modifying the effective deadline of the actuating segment will change the effective deadlines of its corresponding computing and sensing segments. If the modification of a candidate actuating segment incurs no ECOIs or EVOIs, the modification is accepted. If there exists no feasible candidate segment, the algorithm returns failure.

The key design principle of the above greedy heuristics is based on the observation that a segment with an earlier effective release time is usually preempted by the segments with later effective release times. The proposed adjustment improves the priority of the current candidate segment. In addition, since the effective release time of its preceding segment(s) is also modified, it makes more space for scheduling the current candidate segment. The segment with the earliest effective release time is preempted by the largest amount of other segments, thus the modification enables the effective adjustment of the priority to the utmost. This modification procedure repeats until EDF finds feasible schedules. Each segment can be modified at most $O(N)$ time and there are at most $O(N)$ number of segments. Since for each round of modification, identifying the overload/tight intervals takes $O(N^2)$ time, the overall time complexity of Alg.~\ref{alg4} is $O(N^4)$, where $N=\sum_{i=1}^{n} \mathcal{H}/T_i$ is the total amount of instances of all the tasks.

\section{Performance Evaluation}\label{sec:evaluation}

In this section, we evaluate the performance of the proposed algorithms for solving the CRS problem under the $h\text{-}1\text{-}1$, $1\text{-}m\text{-}1$, and the general model. The proposed CRS algorithms are compared with the baseline algorithms, including both EDF and least laxity first (LLF). For EDF, we use two ready queues to separately store the network and computing segments. A segment is popped out from the queue if its corresponding task has the earliest deadline. A segment is released and pushed into the ready queue when its preceding segment is finished. Similarly, LLF also uses two ready queues. In each queue, the priority of the segments is defined as the laxity of the tasks, i.e. the task deadline minus the remaining execution time.     
We use NEC to denote the number of task set satisfying the necessary condition of finding a feasible schedule, which is the upper bound of the feasible cases under the general model. We search ECOI and ENOI for each trial of task set. If any ECOI or ENOI is identified, the trial will be recorded as an infeasible case and will not be counted in NEC. 


\subsection{Experimental Setup} {\label{simu:setup}}

To efficiently obtain the feasible solution of constraint programming, we employ an efficient satisfiability modulo theories (SMT) solver Z3. All the algorithms including SMT are implemented in Python and computed in a CPU cluster node with Xeon E5-2690 v3 2.6 GHz CPU. To perform an extensive comparison, we generate 1000 trials under each parameter setting. Each trial contains a task set $\mathcal{T} = \{\tau_1, \tau_2, ..., \tau_n\}$, where $n \in [1, 50]$. The task periods are randomly set in $[10, 10000]$. For the task set generated under the $h\text{-}1\text{-}1$ model, since its network resource utilization is much larger than the computing resource utilization, we vary the network resource utilization to assess the impact. Similarly, we vary the computing resource utilization of the task set generated under the $1\text{-}m\text{-}1$ model. For the task set under the general model, we use the normalized resource utilization to represent the resource utilization of a task, where the normalized utilization is  calculated as  $(C_{i,s}+C_{i,c}+C_{i,a})/2T_i$ (we use $2T_i$ because there are two resources, CPU and network bandwidth). Given the utilization of a task set, we generate the utilization of all the tasks using the UUniSort algorithm \cite{bini_measuring_2005}. For the general model, we randomly split the total execution time into the sensing, computing, and actuating segments.

\vspace{-0.05in}
\subsection{Evaluation Results} {\label{simu:results}}

In the first set of experiments, we compare the performance of the proposed solutions with the baseline methods for the general model, $h\text{-}1\text{-}1$ model, and $1\text{-}m\text{-}1$ model by varying the resource utilization of the task set. Fig.~\ref{fig:simu_h_m_k} shows the percentage of feasible cases obtained by CRS, EDF, and LLF under the general model when the normalized resource utilization is varied. One can observe that as the normalized resource utilization increases from $0.2$ to $0.9$, the percentage of feasible cases that NEC, CRS, EDF, and LLF can achieve gradually decrease from $96.6\%$, $96.2\%$, $94.7\%$, and $93.2\%$ to  $1.3\%$, $0.3\%$, $0.1\%$, and $0.1\%$, respectively. Fig.~\ref{fig:simu_h_m_k} also shows that the performance of the proposal heuristic solution is close to the upper bound (the difference increases from 0.4\% to 7\% when the utilization is increased from 0.2 to 0.9), and outperforms those of EDF and LLF significantly. For example, it reports more than $10\%$ and $7\%$ feasible cases than EDF and LLF, respectively, when the utilization is set at 0.6. 

The performance of the algorithms under $h\text{-}1\text{-}1$ model is shown in Fig.~\ref{fig:simu_h_1_1}. It is observed that the percentage of feasible cases that CRS, EDF, and LLF can achieve gradually decrease from $97.3\%$, $96.8\%$, and $91.0\%$ to $10.9\%$, $8.4\%$, and $7.8\%$, respectively. Our method outperforms EDF and LLF when the utilization is high. For example, it reports $4.5\%$ and $12\%$ more feasible cases than EDF and LLF, respectively, when the utilization is set at 0.7.     
We also evaluate the performance of the proposed algorithm by varying the computing resource utilization of the task set under $1\text{-}m\text{-}1$ model. Fig.~\ref{fig:simu_1_m_1} reports the percentage of feasible schedules obtained by CRS, EDF, and LLF, which decrease from $97.5\%$, $97.5\%$, and $89.9\%$ to  $20.0\%$, $18.6\%$, and $15.6\%$, respectively. CRS shows slightly better performance than EDF due to a small number of cases that identify the virtual network overload intervals. Note that we don't show the upper bound (NEC) for $h\text{-}1\text{-}1$ model and $1\text{-}m\text{-}1$ model because the proposed method is the optimal method.

\begin{figure}[t!] 
\centering

    \subfloat
    [General model]
    {\label{fig:simu_h_m_k}
    \includegraphics[width=3.3in]{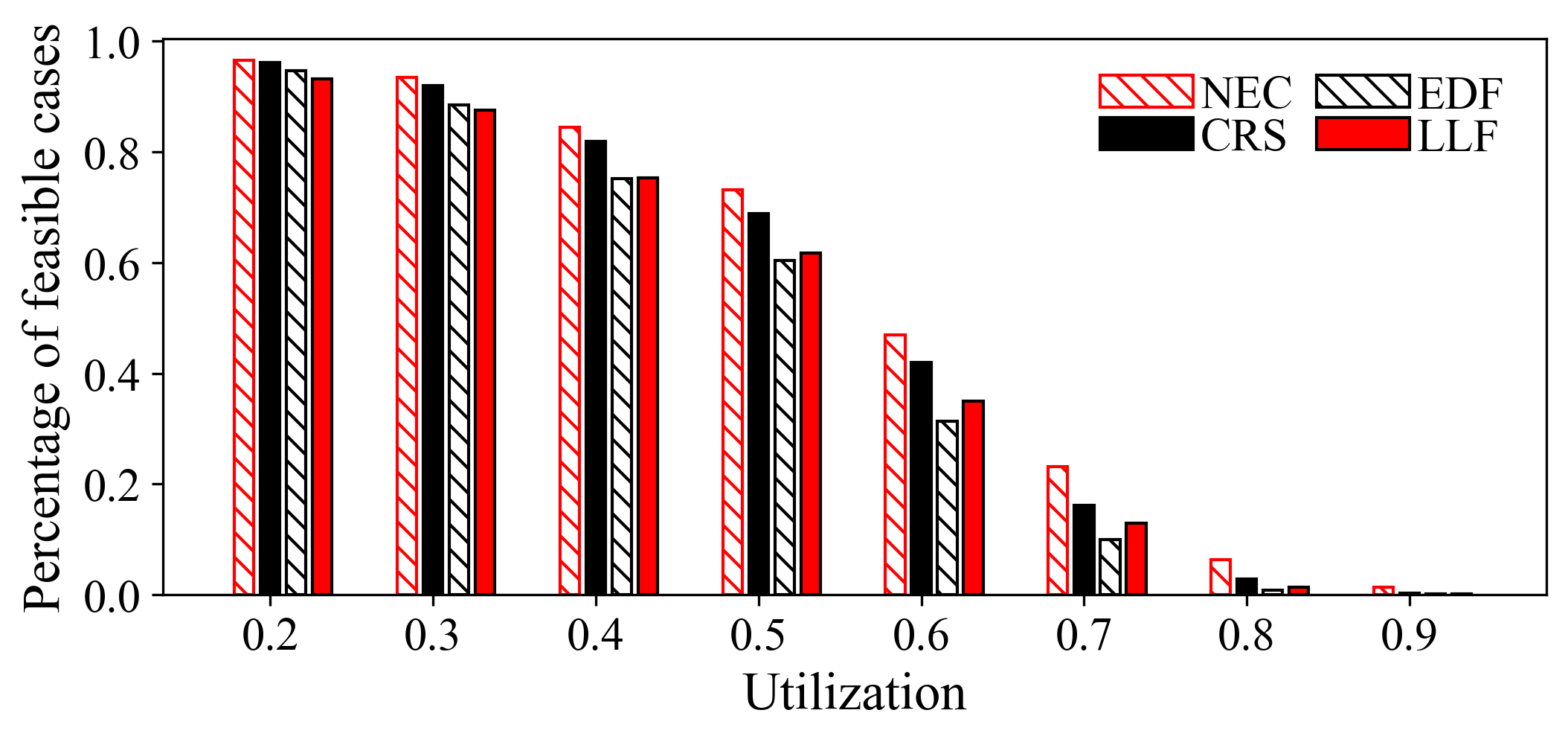}}
     \hspace{0.005\textwidth}
    \subfloat
    [$h\text{-}1\text{-}1$ model]
    {\label{fig:simu_h_1_1}
    \includegraphics[width=3.3in]{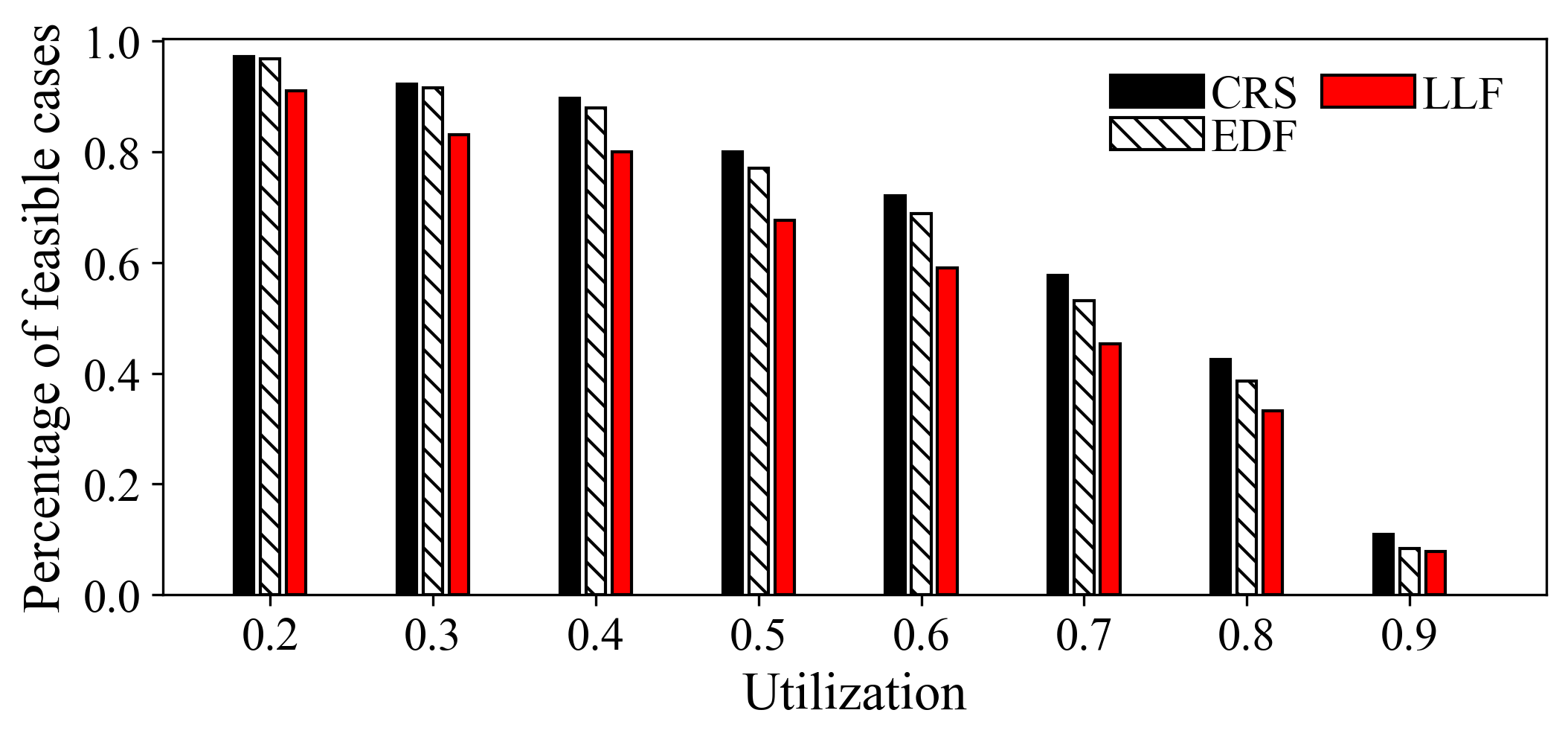}}
     \hspace{0.005\textwidth}
    \subfloat
    [$1\text{-}m\text{-}1$ model]
    {\label{fig:simu_1_m_1}
    \includegraphics[width=3.3in]{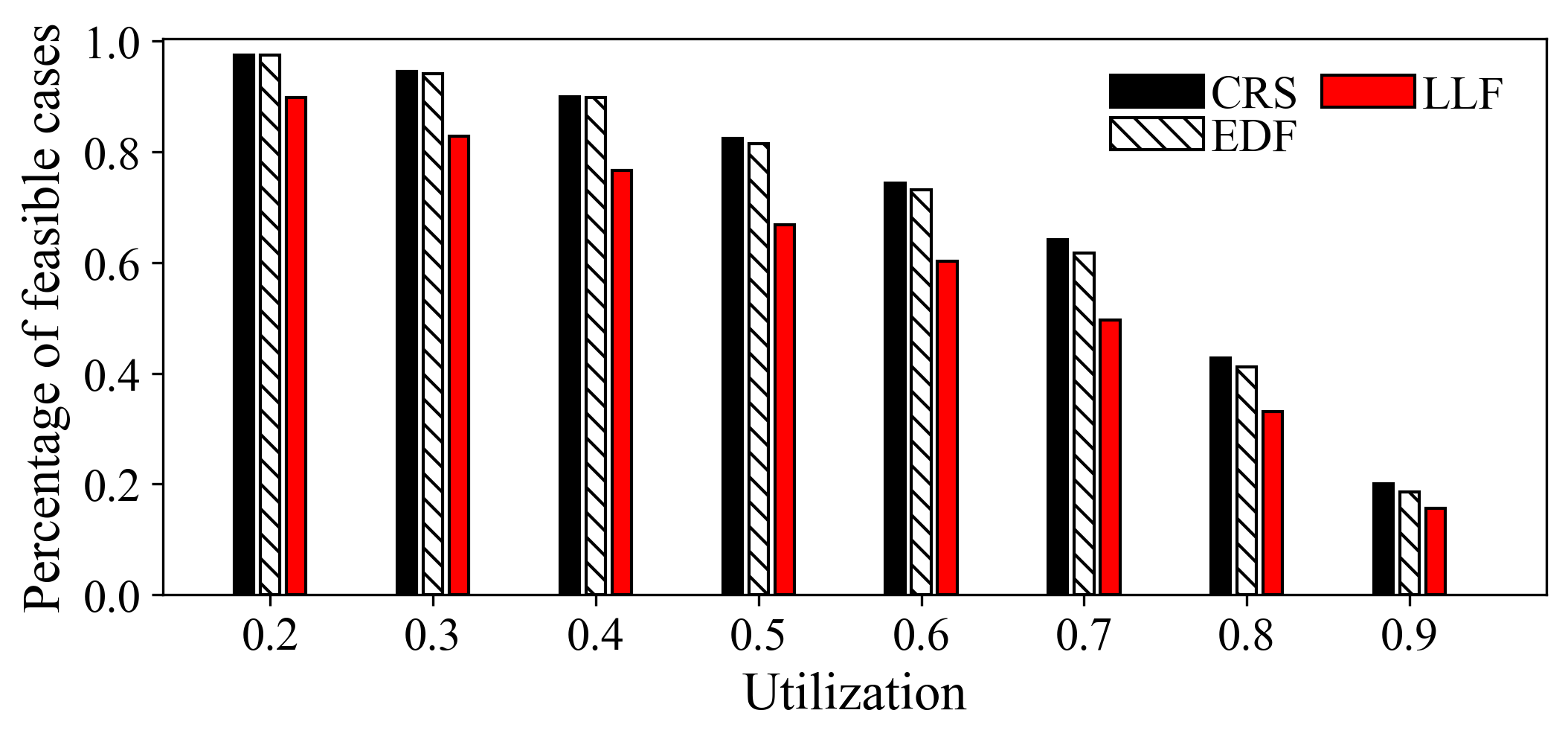}}
     \hspace{0.005\textwidth}
    
\caption{Percentage of feasible cases versus utilization under different models}
\label{fig:simu_models}
\end{figure}

\eat{
\begin{figure}[t!] 
\centering
    \subfloat
    [General model]
    {\label{fig:simu_h_m_k}
    \includegraphics[width=3.3in]{plot_norm_hmk_u_bound.png}}
     \hspace{0.005\textwidth}
    \subfloat
    [$h\text{-}1\text{-}1$ model]
    {\label{fig:simu_h_1_1}
    \includegraphics[width=3.3in]{plot_norm_m11_u_bound.png}}
     \hspace{0.005\textwidth}
    \subfloat
    [$1\text{-}m\text{-}1$ model]
    {\label{fig:simu_1_m_1}
    \includegraphics[width=3.3in]{plot_norm_1m1_u_bound.png}}
     \hspace{0.005\textwidth}
    
\caption{Percentage of feasible cases versus utilization under different models}
\label{fig:simu_models}
\end{figure}
}

In the second set of experiments, we further evaluate the performance of the SMT method for the general model. The results are shown in Table~\ref{tab:comp_SMT}. In the experiments, we generate 1000 cases for the evaluation. Since SMT may take a long time to obtain a solution in some corner cases, we set the longest task period to 1000s in each case and the average number of jobs is 100 for the 1000 cases. We also set a time limit of $2$ hours for running SMT. When the time limit is reached, SMT terminates and that trial will be recorded as an unsolved case. From Table~\ref{tab:comp_SMT}, it is observed that the average running time of SMT for the solved cases is 1940s and the percentage of solved cases of SMT is only 15.1\%. Compared with SMT, our method can finish a trial in 74 milliseconds on average and can solve all the 1000 cases.

\begin{table}[h]
    \centering
    \small
    \caption{Performance comparison with SMT}
    \begin{tabular}{|c|c|c|c|}
    \hline 
     & \multicolumn{2}{c|}{Average running time} & Percentage of terminated\\
     \cline{2-3}
     & Unsolved & Solved & cases due to time limit\\

        \hline  
        SMT & $>$2h & 1940s & 84.9\%
         \\
         CRS & N/A & 0.074s & 0\%
         \\
         EDF & N/A & 0.008s & 0\%
         \\
       LLF & N/A & 0.008s & 0\%
         \\
    \hline
    \end{tabular}
    \label{tab:comp_SMT}
\end{table}

In the last set of experiments, we evaluate the running time of the algorithms by increasing the number of jobs from 100 to 10000. The comparison results are shown in Fig.~\ref{fig:simu_job_time}. It can be observed that the average running time of all the algorithms increase along with the increase of the number of jobs. For all the models, the running time of CRS are under 20s when the number of jobs are less than 2000. When the number of jobs reaches 10,000, the running time of CRS is under 500s, which is acceptable compared with the running time of EDF and LLF methods. SMT terminates due to the time limit (7200s) when the number of jobs are larger than 1000. Note in Fig.~\ref{fig:simu_job_time}, we only show the running time of the feasible cases for CRS, EDF, and LLF. For the infeasible cases, the average running time is under 10s.

\begin{figure}[t!] 
    \centering
    
    \subfloat
    [General model]
    {\includegraphics[width=3.1in]{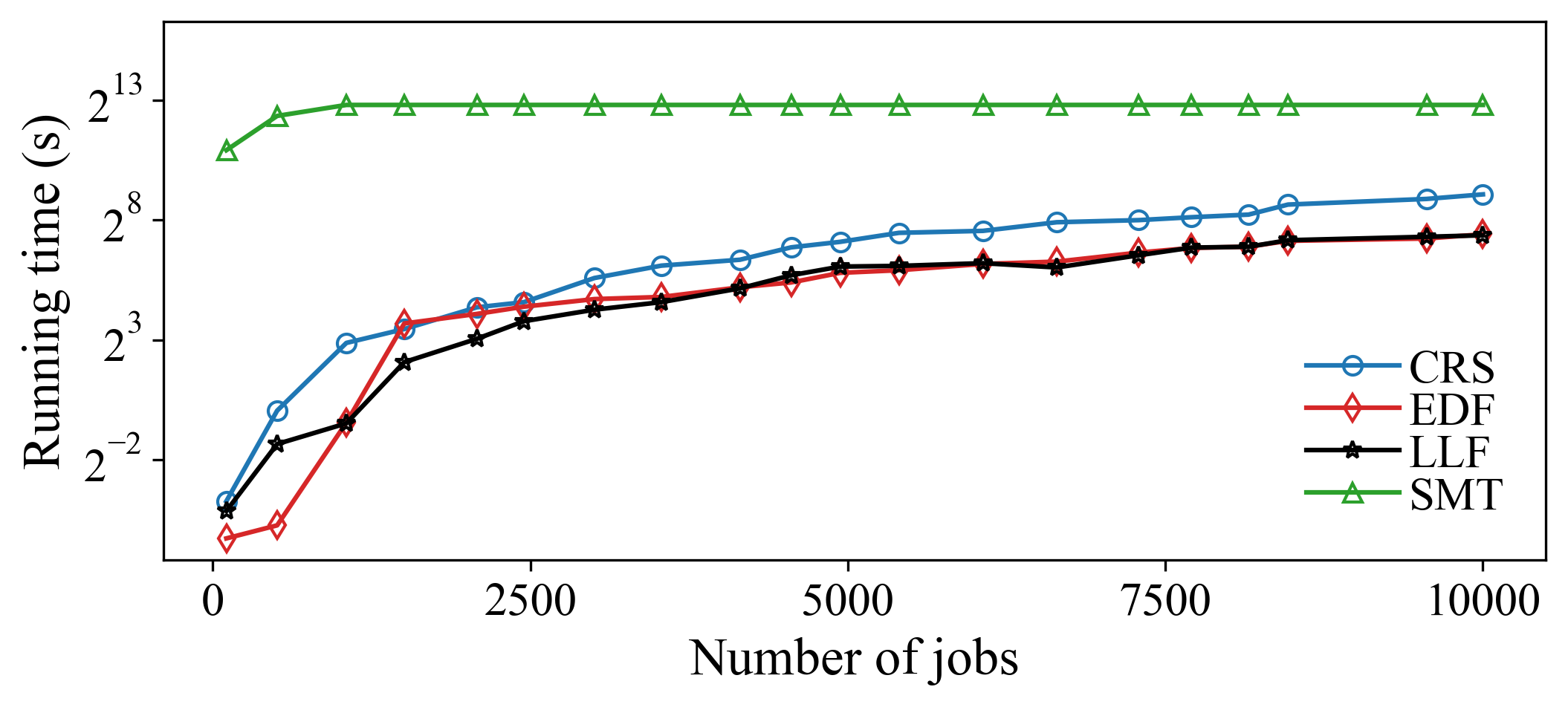}}
     \hspace{0.005\textwidth}\\
    \vspace{-0.1in}
    \subfloat
        [$h\text{-}1\text{-}1$ model]
        {\includegraphics[width=1.55in]{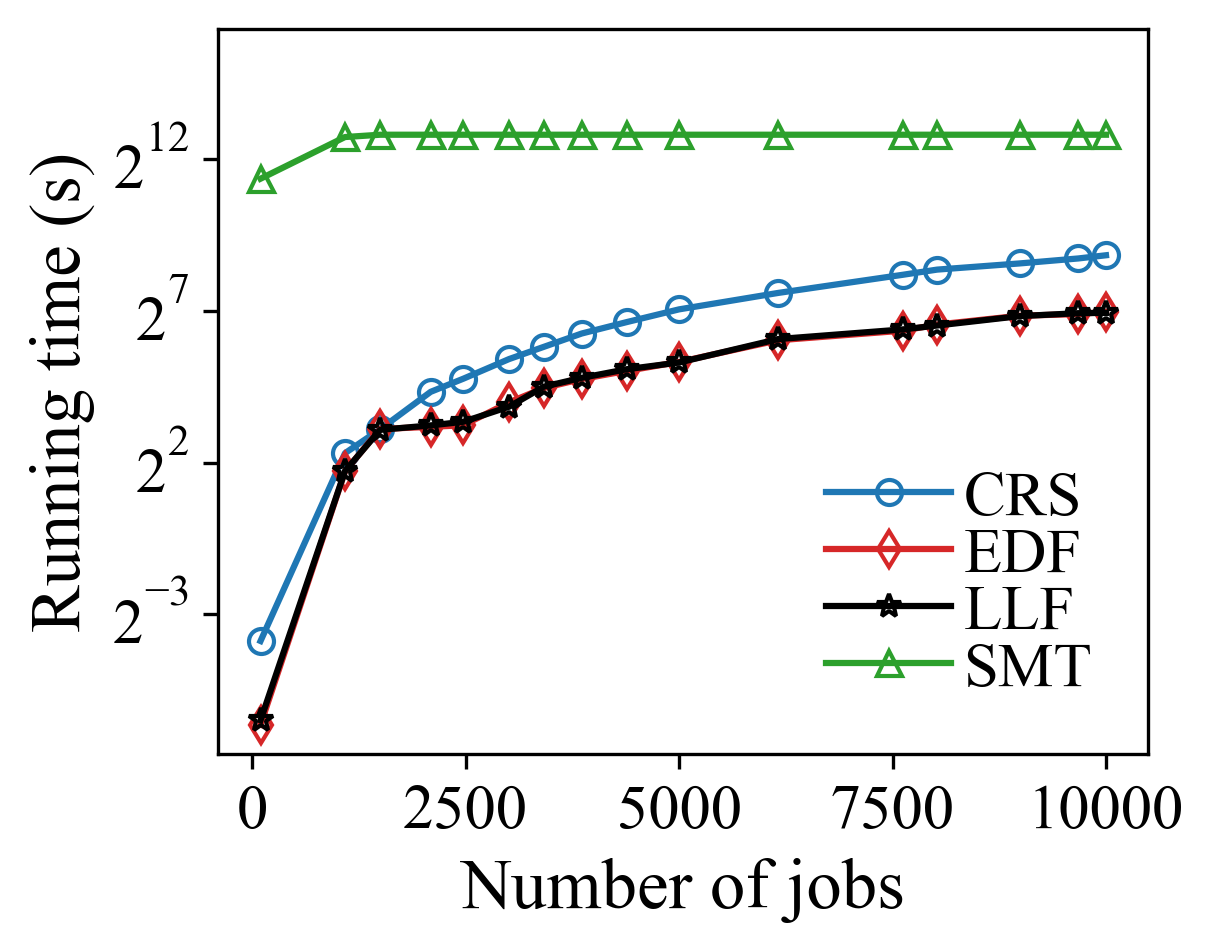}}
         \hspace{0.005\textwidth}
    \subfloat
        [$1\text{-}m\text{-}1$ model]
        {\includegraphics[width=1.45in]{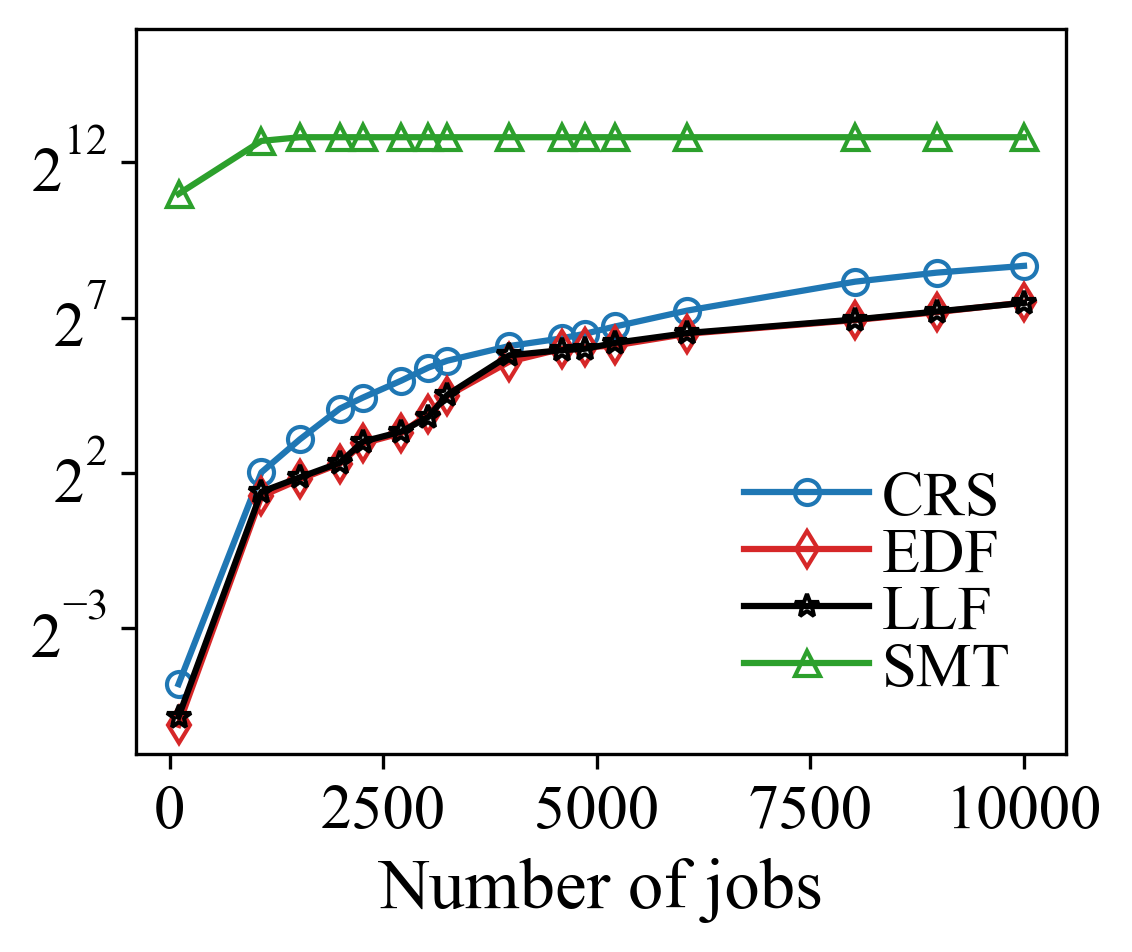}}
        \hspace{0.005\textwidth}
        
    \caption{Comparison of running time with different number of jobs}
    \label{fig:simu_job_time}
\end{figure}

\eat{
Fig~\ref{fig:simu_special_311} shows.......
\begin{figure}[t!] 
    \centering
    \subfloat
    [$h\text{-}1\text{-}1$ model (plot 1)]
    {\includegraphics[width=3.3in]{plot_311_u.png}}
     \hspace{0.005\textwidth}\\
    \vspace{-0.1in}
    \subfloat
        [$h\text{-}1\text{-}1$ model (plot 2)]
        {\includegraphics[width=3.3in]{plot_norm_311_u_bound.png}}
         \hspace{0.005\textwidth}
        
    \caption{Average running time of different number of jobs for special case m11(m=3)}
    \label{fig:simu_special_311}
\end{figure}
}
\section{Related Work}\label{sec:relatedWorks}

Most real-time networks adopt Time Division Multiple Access (TDMA) based data link layers to guarantee deterministic real-time communications. Sensing and actuating tasks are abstracted as end-to-end flows with specified timing requirements. The existing real-time network scheduling algorithm designs focus on schedulability analysis and management of the network packet scheduling (e.g.,~\cite{sisinni_industrial_2018,saifullah_real-time_2010,saifullah_end--end_2015, zhang_distributed_2017, gong_reliable_2019, zhang_fully_2021}). Those solutions may fit well for NCSs when the transmission time is large 
and the computation time of the controller is negligible. However, 
with the rapid development of real-time network such as RT-WiFi~\cite{wei_rt-wifi_2013}, which supports a minimal time slot of $0.1$ milliseconds, the transmission time of the network segment becomes comparable to the computation time of controller~\cite{chen_online_2019}. In addition, due to an increasing number of industrial applications having complex online optimization algorithms designed for the controllers, the computation time can no longer be ignored~\cite{jo_development_2014, jo_development_2015}. In this paper we propose a three-segment execution model for scheduling composite network and CPU resources in NCSs. To provide more flexibility,  both network segments and computing segments are considered to be preemptive.

The existing related three-segment execution models include self-suspension model, PRedictable Execution Model (PREM) and Acquisition-Execution-Restitution (AER) model. The self-suspension model is composed of two computation segments separated by one suspension interval, which focuses on one resource type~\cite{chen_scheduling_2019}. PREM and AER are designed for multi-core CPU system. PREM enables parallelism by dividing tasks in communication/computation phases ~\cite{alhammad_time-predictable_2014}. The read phase reads data from the main memory, the computation phase can proceed the execution, and the write phase writes the resulting data back to the main memory. Since PREM increases the predictability of an application by isolating memory accesses, it is widely used~\cite{becker_contention-free_2016, becker_scheduling_2018, wasly_hiding_2014}. However, the PREM model usually couples the write phase of a task with the next activated read phase on the same core~\cite{wasly_hiding_2014, alhammad_memory_2015}. By contrast, AER from ~\cite{maia_closer_2016} allows more freedom to schedule the read and write phases. To implement the AER model on a scratchpad memory (SPM) based single-core, the SPM is divided into two parts for the currently executing task and the next scheduled task~\cite{wasly_hiding_2014}. Hiding the communication latency at the basic-block level due to a modification of the LLVM compiler toolchain was introduced and a co-scheduling and mapping of computation and communication phases from task-graph for multi-cores was presented in~\cite{soliman_wcet-driven_2017}. An offline scheduling scheme for flight management systems using a PREM task model is presented in~\cite{durrieu_predictable_2014}, which avoids interferences to access the communication medium in the schedule. In addition, the study in~\cite{rouxel_hiding_2019} proposed to minimize the communication latency when scheduling a task graph on a multi-core by overlapping communications and computation based on the AER model. Similarly, to allow overlapping the memory and execution phases of segments of the same task, a new streaming model was introduced in~\cite{soliman_segment_2019}. To improve the schedulability of latency-sensitive tasks, a protocol that allows hiding memory transfer delays while reducing priority inversion was proposed in~\cite{casini_predictable_2020}. A memory-centric scheduler for deterministic memory management on COTS multiprocessor platforms without any hardware support was proposed in~\cite{schwaricke_fixed-priority_2020}. To increase system determinism by reducing task switching overhead, an Integer Linear Programming (ILP) formula optimization method without contention was proposed in~\cite{koike_contention-free_2020}. The work in~\cite{schuh_study_2020} studied global static scheduling, time-triggered and non-preemptive execution of tasks by implementing SCADE applications. As those recent works above focus on scheduling for multi-core CPU systems, the communication phases are reasonably considered to be non-preemptive, while the network scheduling in our CRS model employs the preemptive segments to provide better parallelism.

For the existing related works on end-to-end scheduling in NCSs, utilizing the system state of  control systems to design the scheduler is an important methodology to optimize QoS. Based on the system state, the scheduling and feedback co-design for NCSs is introduced in~\cite{branicky_scheduling_2002, saha_dynamic_2015}, which computes the deadline for the real-time tasks. However, the network scheduling is not guaranteed to be real-time because of the CAN protocol considered in the model. The work in~\cite{lesi_integrating_2020} studies how to integrate security guarantees with end-to-end timeliness requirements for control tasks in resource-constrained NCSs. The proposed sensing-control-actuation model is similar to our CRS model, but the sensing, computing and actuating segments in the proposed model have designed release times and deadlines. This is not general as the CRS model which only has a task deadline. Other related network and computing co-scheduling works include the co-generation of static network and task schedules for distributed systems consisting of preemptive time-triggered tasks which communicate over switched multi-speed time-triggered networks, which is solved by a SMT solver~\cite{craciunas_combined_2016}, and the feasible time-triggered schedule configuration for control applications, which was designed to minimize the control performance degradation of the applications due to resource sharing~\cite{minaeva_control_2020}.

\section{Conclusion and Future Work}\label{sec:conclusion}

In this paper, we study the composite resource scheduling problem in networked control systems (NCSs). A new composite resource scheduling (CRS) model is introduced to describe the sensing, computing, and actuating segments in NCSs. We formulate the composite resource scheduling problem in constraint programming and prove it to be NP-hard in the strong sense. Two special models and the general model are studied. For the CRS problem under the $h\text{-}1\text{-}1$ model, we present an optimal algorithm that utilizes the intervals of network resource utilization of $100\%$ to prune the search space to find the solution. For the CRS problem under the $1\text{-}m\text{-}1$ model, we propose an optimal algorithm that exploits a novel backtracking strategy to adjust the timing parameters of the tasks so that there exist no intervals of network resource utilization larger than $100\%$ obtained in the two-stage decomposition method. For the general case, we propose a heuristic solution to find the feasible schedule based on the greedy strategy of modifying the segments in the interval of either network resource utilization or computing resource utilization larger than $100\%$. 

As the future work, the proposed algorithms will be implemented on our NCS testbed to evaluate their effectiveness and practicability in real-life systems.
\section{acknowledgement}\label{sec:ack}

The  work  reported  herein  is  partly supported  by the National Science Foundation under NSF Awards CNS-2008463 and IIS-1718738.

\vspace{0.1in}

\bibliographystyle{IEEEtran}
\bibliography{MyLibrary2}

\end{document}